\documentclass[10pt]{amsart}
\pdfoutput=1

\usepackage{amssymb}
\usepackage{amsthm}
\usepackage{amsfonts}
\usepackage{amsmath}
\usepackage{pmboxdraw}
\usepackage{verbatim} 
\usepackage{graphicx}
\usepackage{color}
\usepackage[colorlinks=true, citecolor=blue, filecolor=black, linkcolor=black, urlcolor=black]{hyperref}
\usepackage{cite}
\usepackage[normalem]{ulem}
\usepackage{subcaption}
\usepackage{bbm}
\usepackage{bm}
\usepackage{mathtools}
\usepackage{todonotes}
\usepackage{kantlipsum}
\usepackage{colortbl}
\usetikzlibrary{decorations, arrows.meta, patterns}

\newcommand{\RN}[1]{%
	\textup{\uppercase\expandafter{\romannumeral#1}}%
}

\def\C{\mathbb{C}}

\def\R{\mathbb{R}}

\newcommand{\re}{\operatorname{Re}}
\newcommand{\im}{\operatorname{Im}}

\theoremstyle{plain}
\newtheorem{thm}{Theorem}[section]

\newtheorem{cor}[thm]{Corollary}
\newtheorem{conj}[thm]{Conjecture}
\newtheorem{lem}[thm]{Lemma}

\newtheorem{prop}[thm]{Proposition}

\theoremstyle{remark}
\newtheorem{defn}{Definition}

\newtheorem{rem}{Remark}

\numberwithin{equation}{section}

\begin{document}

\title[Equilibrium measures for rotationally symmetric Riesz gases]{Equilibrium measures for higher dimensional \\ rotationally symmetric Riesz gases}
\author{Sung-Soo Byun}
\address{Department of Mathematical Sciences and Research Institute of Mathematics, Seoul National University, Seoul 151-747, Republic of Korea}
\email{sungsoobyun@snu.ac.kr}

\author{Peter J. Forrester}
\address{School of Mathematical and Statistics, The University of Melbourne, Victoria 3010, Australia}
\email{pjforr@unimelb.edu.au} 

\author{Satya N. Majumdar} 
\address{LPTMS, CNRS, Univ. Paris-Sud, Université Paris-Saclay, 91405 Orsay, France}
\email{satyanarayan.majumdar@cnrs.fr} 

\author{Gregory Schehr} 
\address{Sorbonne Universit\'e, Laboratoire de Physique Th\'eorique et Hautes Energies, CNRS UMR 7589, 4 Place Jussieu,
75252 Paris Cedex 05, France}
\email{schehr@lpthe.jussieu.fr} 


\begin{abstract} 
We study equilibrium measures for Riesz gases in dimension~$d$ with pairwise interaction kernel $|x-y|^{-s}$, subject to radially symmetric external fields. 
We characterise broad classes of confining potentials for which the equilibrium measure is supported on the unit ball and admits an explicit density.
Our main contribution is a converse construction: starting from a prescribed radially symmetric equilibrium density given as a power series in the squared radius, we determine the associated external potential and establish the corresponding Euler--Lagrange variational conditions. 
A key ingredient in the proof is an identity between two ${}_3F_2$ hypergeometric functions evaluated at unit argument, which is of independent interest.
As applications, we identify the external potentials corresponding to equilibrium densities proportional to $(1-|x|^2)^\alpha$, $\alpha>-1$, and show that these potentials can be expressed in terms of Gauss hypergeometric functions ${}_2F_1$, reducing to polynomials for special values of~$\alpha$. 
We also determine the equilibrium measure associated with purely power-type external potentials---often referred to as Freud or Mittag--Leffler potentials in the context of log gases---for which the equilibrium density admits an explicit ${}_2F_1$ representation.
Furthermore, we apply our framework to a Coulomb gas in dimension $d+1$ confined by a harmonic potential to the half-space. We derive a necessary condition under which the equilibrium measure is fully supported on the boundary hyperplane of dimension~$d$, with the induced density corresponding to that of a Riesz gas with exponent $s=d-1$.
\end{abstract}


\maketitle

\section{Introduction}

Long-range repulsive interactions play a central role in many-body systems and have therefore attracted a wide interest across both physics~\cite{CDFR14} and mathematics~\cite{ST97,Lew22,Se24,BF25a}. 
When combined with an external confining field, such interactions typically give rise, in the thermodynamic limit, to equilibrium configurations with compact support, nontrivial density profiles, and microscopic structures that depend on the form of the interactions. These features appear in a broad range of physical contexts, including classical plasmas and electrolytes~\cite{Bax63,AM1980,CKMN81,JLM93}, trapped ions and cold atomic gases~\cite{BC03}, as well as rotating fermionic systems and effective descriptions of quantum Hall states~\cite{CE2020,Coo10,OLMSE24,LMS19,SDMS22}.

A unifying theoretical framework for these systems is provided by \emph{Riesz gases}~\cite{Riesz38}, which describe $N$ identical particles in dimension $d$ interacting through a repulsive pairwise potential of the form $|x_i-x_j|^{-s}$, with $x_i$'s denoting the positions of the particles. The exponent $s$ governs the effective interaction range. 
By varying $s$, Riesz gases interpolate between Coulomb systems ($s=d-2$), logarithmic interactions ($s\to 0$) which play a central role in random matrix theory~\cite{Meh04,Fo10,BF25}, and more general power-law interactions. The interplay among the interaction exponent $s$, the spatial dimension $d$, and the confining potential determines both the shape of the equilibrium measure and the structure of its support, thereby offering a flexible framework to investigate how dimensionality and interaction range influence collective behaviour. In particular, these systems exhibit qualitatively distinct regimes depending on whether the interactions are genuinely long-ranged ($s<d$) or effectively short-ranged ($s>d$).

While many general properties of Riesz gases--such as the existence of equilibrium measures, large-deviation principles, and qualitative features of minimisers--are now well understood~\cite{Lew22,Se24}, explicit descriptions of equilibrium densities remain scarce. In one dimension, interesting progress has been made recently, in particular in the physics literature: for the harmonic potential \cite{ADKKMMS19}, and in the presence of an additional hard wall \cite{DKMSS17,KKKMMS21,KKKMMS22,DM06,DM08}, the equilibrium measure can be computed explicitly, often in connection with the statistics of the rightmost particle; see also \cite{LS25} in the context of linear and counting statistics. In higher dimensions, however, and for general confining potentials, such explicit results are largely unavailable, leaving many questions about equilibrium structures open. Even for radially symmetric potentials, which is the setting we focus on in this paper, closed-form expressions are typically available only in special situations, such as harmonic confinement or for specific interaction exponents.
Notable examples include the Coulomb gas $s=d-2$, as well as more exotic
``iterated Coulomb gases'' with $s=d-4$ \cite{CSW23}; see also \cite{CSW22} for $s=d-3$. 
In this work, we study rotationally invariant Riesz gases in arbitrary dimension under radially symmetric external potentials and identify broad
classes of equilibrium densities that admit explicit expressions.
 
The Riesz gas studied in this paper is defined as follows. For positive integers $N$ and $d$, let $X_N=(x_1,\ldots,x_N)\in(\mathbb{R}^d)^N$ be a point configuration.  
We consider the associated $N$–particle Hamiltonian
\begin{equation}  \label{def of Hamiltonian}
\mathcal{H}_N(X_N) := \sum_{1 \le i \ne j \le N} g(x_i - x_j) 
+ N \sum_{i=1}^N V(x_i),
\end{equation}
where $V:\mathbb{R}^d\to(-\infty,+\infty)$ is a confining potential and  
\begin{equation}
\label{eq:g}
g(x)  := \begin{cases}
\frac{1}{s} |x|^{-s}, & s \ne 0, 
\smallskip 
\\
-\log|x|, & s = 0 , 
\end{cases} \qquad x \in \R^d. 
\end{equation} 
Without loss of generality, we often normalise the confining potential by $V(0)=0.$ 
In particular, the choice $s=d-2$ yields the $d$–dimensional Coulomb gas. Note that the case $s=0$, corresponding to log gases \cite{Fo10}, can be formally recovered from the general Riesz gas by taking the limit since
\begin{equation}
-\log |x|= \lim_{s\to 0} s^{-1}( |x|^{-s}-1 ) .
\end{equation} 
A special feature of the Coulomb potential is that it satisfies the $d$-dimensional Poisson equation
\begin{equation}\label{1.3a}
\Delta g(x) |_{s=d-2} = - c_d \delta(x), \qquad c_d = 2 \pi^{d/2}/ \Gamma(d/2).
\end{equation}
Here $\Delta$ is the Laplacian, $\delta(x)$ is the $d$-dimensional Dirac delta, and $c_d$ is the surface area of the unit ball in $\mathbb R^d$.
We consider the canonical Gibbs measure at inverse temperature $\beta>0$, given by 
\begin{equation}
d \mathbb{P}_{N,\beta}(X_N) = \frac{1}{Z_{N,\beta}} \exp\Big(-\frac{\beta}{2} N^{-\frac{s}{d}} \mathcal{H}_N(X_N) \Big) dX_N. 
\end{equation}
Here we adopt a scaling under which the resulting equilibrium measure is supported on a compact set. Other conventions are also commonly used in the literature, in which the support grows with~$N$; see e.g. \cite{ADKKMMS19}.

By definition, the equilibrium measure $\mu_V$ is the unique minimiser of the energy 
\begin{equation} \label{def of energy functional}
I_V[\sigma] := \iint_{ (\R^d)^2} g(x-y) \,d\sigma(x)\,d\sigma(y) + \int_{ \R^d } V(x)\,d\sigma(x),  
\end{equation}
subject to the constraint that the total charge $Q=\int \sigma(x) \, dx$ is conserved.
Furthermore, according to Frostman's theorem \cite{Fr35} (see also \cite[Theorem~2]{DOSW23}), for $0 < s < d$, $\mu_V$ is characterised by the Euler–Lagrange condition 
\begin{equation} \label{def of EL eqn}
2\int_{ \R^d } g(x-y) \,d\mu_V(y) +V(x) \begin{cases}
= c & \textup{if } x \in \textup{supp}(\mu_V),  
\smallskip 
\\
\ge c & \textup{if } x \not\in \textup{supp}(\mu_V) \;.  
\end{cases} 
\end{equation} 
Here the constant $c$ is called the (modified) Robin's constant.  
 
In this work, we aim to find a class of potentials $V$ and its associated equilibrium measure supported on the unit ball $\{x\in \R^d: |x| \le 1 \}$ for which both admit an explicit expression.
Our starting point is to prescribe the equilibrium measure and then determine the corresponding potential. Specifically, upon a minor abuse of notation we write
$d \mu(x) = \mu(x) \, dx$, where
\begin{equation}\label{def of eq msr mu ak gen d}
\mu(x) = {d \over  c_d} \bigg( \sum_{k=0}^\infty a_k |x|^{2k} \bigg) \cdot  \mathbbm{1}_{ \{ |x| \le 1 \} } \,.
\end{equation}
Here, $c_d$ is given as in  (\ref{1.3a}).
 Assuming that $\mu$ defines a valid probability density, we then ask: 
\smallskip 
\begin{center}
\textit{What is the confining potential $V(x)$ for which $\mu$
arises as the corresponding equilibrium measure?}
\end{center} 
\smallskip 
We aim to contribute to this question for a broad class of external potentials.

Before addressing the general Riesz gas case, we first focus on the Coulomb interaction. In the Coulomb case $s=d-2$, the equilibrium problem admits a particularly transparent formulation. 
The external potential is characterised by the Poisson equation
(\ref{1.3a}) with $g(x) |_{s=d-2}$ replaced by ${1 \over 2} V(x) |_{s=d-2} $ and
the Dirac delta replaced by $-\mu(x)$ as given by (\ref{def of eq msr mu ak gen d}). This reads
\begin{align}\label{e.2}
\Delta \Big( V(x)|_{s=d-2}\Big) = 2d \sum_{n=0}^\infty a_n |x|^{2n}.  
\end{align}
Recalling that the radial part of the $d$-dimensional Laplacian is
${1 \over r^{d-1}}{d \over dr} r^{d-1} {d \over dr}$ with $(r = |x|),$ 
a power series ansatz for $V(x)$ together with our choice of normalisation $V(0) = 0$ shows that
\begin{align}\label{e.3}
   V(x)|_{s=d-2}= \frac{d}{2} \sum_{n=0}^\infty \frac{a_n}{(n+\frac{d}{2})(n+1)} |x|^{2n+2}. 
\end{align}  
From the perspective of equilibrium measures, the equality part of the 
Euler--Lagrange condition \eqref{def of EL eqn} is already satisfied by the above computation. 
The simplest way to ensure that the inequality part of \eqref{def of EL eqn} is also satisfied is to choose the external potential $V(x)=+\infty$ for $|x|>1$. 
In this case, the potential \eqref{e.3} with a hard-wall constraint outside 
$|x|>1$ indeed gives rise to \eqref{def of eq msr mu ak gen d} as an equilibrium measure; see \cite{CFLV17}.
A natural question, which also arises in applications (see e.g.  \cite{ST97}), is whether defining the external potential $V(x)$ for all 
$x\in\mathbb{R}^d$ via \eqref{e.3} preserves the validity of the 
Euler--Lagrange conditions.
For a radially symmetric density in the Coulomb case, with $Q = \int_{|x| < 1} \mu(x) \, dx$ the mass of the measure, it is a celebrated result of Newton that for the outside region $|x| > 1$, $\int_{\mathbb R^d} g(x-y) \, d \mu_V(y) = Q g(x)$. Using this identity, one finds that the inequality part of the Euler--Lagrange condition \eqref{def of EL eqn} requires that, for $|x|>1$,
\begin{align}\label{e.4}
2 Q g(x) + V(x)|_{s = d - 2} \ge c = 2 \int_{\mathbb R^d} g(y) \, d \mu_V(y) .
\end{align} 

When the Coulomb condition $s=d-2$ is removed, the preceding arguments for 
deriving $V(x)$ on $|x|\le 1$ and for verifying the Euler--Lagrange conditions 
after extending $V(x)$ to all of $\mathbb{R}^d$ are no longer applicable. 
This leads to substantial additional difficulties in the study of general Riesz gases and partly explains why relatively few explicit examples are known; see \cite[Section~2]{CSW23} for several known cases and related 
conjectures.

In this setting, the computation of the external potential requires evaluating a genuinely $d$-dimensional integral. 
However, as we shall see, the assumption of radial symmetry reduces this problem to a one-dimensional integral, albeit with a seemingly complicated integrand. 
Moreover, in order to ensure the inequality part of the Euler--Lagrange 
conditions, one is naturally led to the problem of generalising \eqref{e.4}.

Our main results, which will be presented in detail in the following section, can be summarised as follows.
\begin{itemize}
    \item In Theorem~\ref{Thm_radial}, we formulate the general framework. 
    We show how to construct an external potential~$V$ for which the equality part of the Euler--Lagrange condition holds, and we reformulate the inequality part of the Euler--Lagrange condition as an explicit inequality involving the coefficients defining the equilibrium measure. 
    This general result is discussed in Subsection~\ref{S1.1}. 
    We provide two explicit examples for which the required inequalities can be verified. See Table~\ref{Table_summary} below for a summary. 
    \smallskip 
    \item The first example, presented in Theorem~\ref{Thm_eq msr power type}, concerns external potentials expressed in terms of Gauss hypergeometric functions ${}_2F_1$. 
    In this case, the corresponding equilibrium density is proportional to 
    $(1-|x|^2)^{\alpha}$ with $\alpha>-1$. This example is treated in Subsection~\ref{S1.2}.
    \smallskip 
    \item The second example, presented in Theorem~\ref{Thm_eq msr for power potential}, considers purely power-type external potentials of the form $|x|^{2p}$, where $p$ is a nonnegative integer. 
    In this setting, the density of the equilibrium measure admits an explicit representation in terms of a ${}_2F_1$ function. 
    This case is discussed in Subsection~\ref{S1.3}.
    \smallskip 
    \item Finally, in Theorem~\ref{Thm_Coulomb half space}, we study the Coulomb case with a half-space constraint. 
    By analysing the points at which the Euler--Lagrange inequality is violated, we derive a necessary condition under which the ensemble is fully confined to the boundary hyperplane. 
    This result is presented in Subsection~\ref{S1.4}.
\end{itemize}

At a technical level, the most challenging part of our analysis arises in the soft wall setting, where one must establish the inequality part of the Euler--Lagrange condition~\eqref{def of EL eqn}.
In contrast, the classical result of Riesz \cite{Riesz38} (see Corollary~\ref{Cor_power type with hard edges} below) for the equilibrium measure inside a unit ball in $\mathbb R^d$ with constant external potential makes a hard wall assumption that the potential is infinite outside of the ball. This makes the inequality part of the Euler-Lagrange condition redundant.

\section{Main results}

In this section, we present our main results. 

\subsection{Formulation for general case}\label{S1.1}

To handle the general setting, we first introduce a sequence of real numbers $\{a_n\}_{n=1}^\infty$ subject to the following conditions.

\begin{defn}
For $d \ge 1$, let $\{a_n\}_{n=1}^\infty$ be a sequence of real numbers satisfying
\begin{equation} \label{def of norm admissible an}
\sum_{n=0}^\infty \frac{a_n}{2n+d} = \frac{1}{d},
\end{equation}
and assume that, for every $0 \le x \le 1$, the series
\begin{equation} \label{def of seq conv for admissible an}
\sum_{n=0}^\infty a_n x^{2n} \qquad \textup{and} \qquad  \sum_{n=1}^\infty \frac{ \Gamma(\frac{s}{2}+n) \Gamma( \frac{s-d}{2}+n+1 )  }{ \Gamma(\frac{d}{2}+n) \, n! } \bigg( \sum_{k=0}^\infty \frac{ a_k}{ 2n+s-2k-d } \bigg) x^{2n} 
\end{equation}
are well defined (i.e. convergent), and that the first is non-negative.
Under these conditions, we say that the sequence $\{a_n\}_{n=1}^\infty$ is \textit{admissible}.
\end{defn}

Since $\mu(x)= \mu(|x|)$  in \eqref{def of eq msr mu ak gen d} for all $d \ge 1$, we have   
\begin{align}
\int_{ \R^d } \mu(x)\,dx = \frac{2 \pi^{ d/2 } }{ \Gamma(\frac{d}{2}) } \int_0^\infty  \mu (r)r^{d-1}\,dr; 
\end{align} 
note that on the LHS the notation $dx$ denotes the
$d$-dimensional Lebesgue measure, in keeping with the domain of integration. 
Thus the condition \eqref{def of norm admissible an} ensures that $\mu$ is, if non-negative, a probability measure, by integrating to unity. 

We formulate a class of rotationally symmetric potentials $V$ together with their associated equilibrium measure. Recall that the (rising) Pochhammer symbol is given by $(x)_n=\Gamma(x+n)/\Gamma(x)$. 

\begin{thm}[\textbf{Radially symmetric equilibrium measure}] \label{Thm_radial}
Let $d \ge 1$ and $s \in [d-2,d)$. 
For an admissible sequence $\{a_k\}_{k=1}^\infty$ and $x \in \R^d$ with $|x| \le 1$, define 
\begin{equation} \label{def of V gen d}
V(x) :=   \frac{ 2d   }{ s   }   \sum_{n=1}^\infty \frac{ (\frac{s}{2})_n  ( \frac{s-d}{2}+1 )_n  }{ (\frac{d}{2})_n \, n! } \bigg( \sum_{k=0}^\infty \frac{ a_k}{ 2n+s-2k-d } \bigg) |x|^{2n} . 
\end{equation}  
For $x \in \R^d$, let $\mu(x)$ be given by (\ref{def of eq msr mu ak gen d}).
Then we have the following. 
\begin{itemize}
    \item[\textup{(i)}] \textbf{\textup{(Euler-Lagrange equality)}} For any $|x|<1$, we have 
    \begin{equation} \label{EL eqn for hard wall}
\frac{2}{s} \int_{ \R^d } \frac{ d\mu(y) }{ |x-y|^s } +V(x) =  {2 d \over s} \sum_{k=0}^\infty {a_k \over 2k + d - s}. 
\end{equation} 
    Consequently, $\mu$ is the equilibrium measure associated with the potential 
    \begin{equation} \label{def of V hard gen}
V_{ \rm h }(x) := \begin{cases}
V(x) &\textup{if } |x| \le 1,
\smallskip
\\
+\infty &\textup{if } |x| >1. 
\end{cases}
\end{equation} 
    \item[\textup{(ii)}] \textbf{\textup{(Euler-Lagrange inequality)}} Suppose furthermore that $V$ is well defined for all $x \in \mathbb{R}^d$. 
Then $\mu$ is the equilibrium measure for $V$ if and only if the inequality 
\begin{align} \label{inequality for general}
  \frac{ 2d   }{ s   }  \sum_{n=0}^\infty \frac{ (\frac{s}{2})_n  ( \frac{s-d}{2}+1 )_n  }{  (\frac{d}{2})_n \, n! }   \bigg( \sum_{k=0}^\infty \frac{a_k}{2k+d+2n} \bigg) \, \frac{1}{|x|^{2n+s}}
+V(x) \ge {2 d \over s} \sum_{k=0}^\infty {a_k \over 2k + d - s}  
\end{align} 
holds for every $|x|\ge1$.  
\end{itemize}
\end{thm}
 
In the definition of \eqref{def of V hard gen}, the hard wall is placed at the ``natural'' boundary of the support of the equilibrium measure.
    Such a regime is commonly referred to as ``the soft edge meeting the hard edge''; see e.g. \cite{CK08,AKM19,AKMW20,CW25,CFLV17,CFLV19}.

Note that by \eqref{def of energy functional} and \eqref{def of EL eqn}, we have  
\begin{align} \label{energy in terms of Robin}
I_{V}[ \mu_V ]  = \frac{c}{2} +\frac12 \int_{ \R^d } V(x)\,d\mu_V(x). 
\end{align}
As an immediate consequence of \eqref{energy in terms of Robin}, in the two situations described in Theorem~\ref{Thm_radial}(i) and (ii), the associated energy admits an explicit expression: 
\begin{align}
 \label{energy in main result}
I_V[\mu]= I_{V_{ \rm h }}[\mu]  = \frac{ d }{ s }    \sum_{k=0}^\infty \frac{a_k}{ 2k+d-s } +\frac12 \int_{ |x|<1 } V(x)\,d\mu(x). 
\end{align}

\begin{rem}[Coulomb gases $s=d-2$]
Note that as $s \to d-2$, we have
\begin{align*}
V(x) \sim \frac{ d }{ \Gamma( \frac{s-d+2}{2} )  }   \sum_{n=0}^\infty \frac{ 1 }{ (n+\frac{d}{2}) \, (n+1) } \bigg( \sum_{k=0}^\infty \frac{ a_k}{ 2n+2+s-2k-d } \delta_{k,n} \bigg) |x|^{2n+2} . 
\end{align*}
Noting that $\Gamma( \frac{s-d+2}{2} ) (2n+2+s-2k-d)  \delta_{k,n} \to 2$ in this limit, it follows that $V(x)|_{s=d-2}$ is given by (\ref{e.3}).
Indeed, since 
$$
\Delta \Big( |x|^{2n+2} \Big) = 2(n+1)(2n+d) |x|^{2n}, 
$$
where $\Delta$ is the Laplacian, we have the Poisson equation (\ref{e.3}).
Also, the terms in (\ref{inequality for general}) in this limit can be identified with the corresponding terms in (\ref{e.4}) with $Q=1$.  
\end{rem}

\begin{rem}[Connection to the short-range regime $s \to d$] 
We now take the formal limit $s \to d$, which gives a connection to the short range regime. 
In the limit $s \to d$, we have  
\begin{align*}
V(x) \sim   2   \sum_{n=1}^\infty  \bigg( \sum_{k=0}^\infty \frac{ a_k}{ 2n+s-2k-d } \delta_{k,n} \bigg) |x|^{2n} \sim -\frac{2}{d-s} \sum_{n=1}^\infty a_n |x|^{2n} = \frac{2}{d-s} \bigg( a_0- \sum_{n=0}^\infty a_n |x|^{2n} \bigg). 
\end{align*}
We write  
\begin{align} \label{def of short range lim V}
\widetilde{V}(x) = \lim_{s\to d} \frac{d-s}{2} V(x). 
\end{align}
Then it follows that  
\begin{align}
\mu(x) |_{s=d} =  a_0 -  \widetilde{V}(x) . 
\end{align}  
This is consistent with the explicit formula for the equilibrium measure in the short-range case. More precisely, it follows from \cite{ADKKMMS19,HLSS18} that, in the short-range regime, the equilibrium measure is given by a Thomas--Fermi type density of the form $(c-\widetilde{V}(x))^{d/s}$. 
We note that, by convention, in the short-range regime $s>d$, the prefactor $N$ in front of the potential $V$ in \eqref{def of Hamiltonian} is replaced by $N^{s/d}$. In the critical case $s=d$, however, the definition of the prefactor requires an additional $\log N$ correction; see \cite{HS04,HS05,KS98} as well as Appendix~C of the arXiv version of \cite{ADKKMMS19}. Such logarithmic corrections are consistent with the need for the renormalisation appearing in \eqref{def of short range lim V}.
\end{rem}

\begin{rem}[Reduction via a hypergeometric identity]
In the derivation of the potential \eqref{def of V gen d}---equivalently, in the derivation of \eqref{EL eqn for hard wall}---a significant simplification arises from Lemma~\ref{Lem_algebraic equal}; see the proof of  Proposition~\ref{Prop_evaulation of the EL conditions}. 
The identity in Lemma~\ref{Lem_algebraic equal} can be expressed in terms of generalised hypergeometric functions (see e.g. \cite[Chapter~16]{NIST}) as follows:
\begin{equation} \label{Id_3F2}
 \frac{ 1 }{ 2k+d  }  {}_3F_2(  1+\tfrac{s-d}{2}, \tfrac{d}{2}+k, \tfrac{s}{2};  \tfrac{d}{2}, \tfrac{d}{2}+k+1; 1  ) = \frac{ 1 }{  2k+ d-s   }  {}_3F_2( 1+\tfrac{s-d}{2}, \tfrac{s-d}{2}-k, \tfrac{s}{2};  \tfrac{d}{2}, \tfrac{s-d}{2}+1-k; 1  ),
\end{equation}
valid for $k \in \mathbb Z_{\ge 0}$ and general $s,d$ such that $s \in (d-2,d) $ (this condition being sufficient for the underlying sums to be convergent). A large class of such hypergeometric identities may be viewed in the framework of classical Thomae relations and their extensions; see, for instance, \cite{KR06,KR03}. That said, we have not been able to match (\ref{Id_3F2}) in this class.
\end{rem}

In the following subsections, we present some prominent examples of the equilibrium measures; see Table~\ref{Table_summary} for the summary.

\begin{table}[h!]
    \centering
    \begin{tabular}{ | c | c | c | c | } 
        \hline
        \cellcolor{gray!25}  & \cellcolor{gray!10} \parbox[c]{2.5cm}{\centering \vspace{.5\baselineskip} $a_k$ \vspace{.5\baselineskip} }  & \cellcolor{gray!10} \parbox[c]{4cm}{\centering \vspace{.5\baselineskip} $V(x)$ \vspace{.5\baselineskip}}    & \cellcolor{gray!10} \parbox[c]{4cm}{\centering \vspace{.5\baselineskip} $\mu(x)$ \vspace{.5\baselineskip}} 
        \\
        \hline
        \cellcolor{gray!10} \parbox[c]{3.5cm}{\centering \vspace{.5\baselineskip} Power-type equilibrium measure, cf.  Theorem.~\ref{Thm_eq msr power type} 
        \vspace{.5\baselineskip} } 
        & \eqref{def of ak for power type eq msr} 
        &     $\sum_{k=1}^{2m+1} b_k  |x|^{2k}$, cf.\eqref{def of V poly soft}
   & $  (1-|x|^2)^{ \frac{s-d}{2}+2m+1 }  $
         \\
        \hline
        \cellcolor{gray!10}  \parbox[c]{3.5cm}{\centering \vspace{.5\baselineskip}  Power-type \\  external potential, \\ cf.  Theorem~\ref{Thm_eq msr for power potential} \vspace{.5\baselineskip} }    &  \eqref{def of ak for power type potential}  &  $ |x|^{2p}$  & $  {}_2F_1( \tfrac{d-s}{2}, \tfrac{d-s-2p}{2}; \tfrac{d-s+2-2p}{2}; |x|^2 )$  \\
        \hline
    \end{tabular}
    \caption{The table summarises the examples, listing both the potential 
and the associated equilibrium measure, up to multiplicative constants.} \label{Table_summary}
\end{table}

\subsection{Power type equilibrium measure}\label{S1.2}

As a first prominent example, we consider a class of polynomial potentials whose associated equilibrium measure is of power type. To proceed, let us introduce the corresponding sequence $\{a_k\}$.
We first recall that the hypergeometric function ${}_2 F_1$ is defined by the Gauss series 
\begin{equation} \label{def of 2F1 Gauss series}
	{}_2 F_1(a,b;c;z):=\frac{\Gamma(c)}{\Gamma(a)\Gamma(b)} \sum_{s=0}^\infty  \frac{\Gamma(a+s) \Gamma(b+s) }{ \Gamma(c+s) s! } z^s, \qquad (|z|<1)
\end{equation}
and by analytic continuation elsewhere. Note that if $a=-m$ with $m=0,1,2,\dots$, the hypergeometric function becomes the polynomial 
\begin{equation} \label{def of 2F1 Gauss series poly case}
	{}_2 F_1(-m,b;c;z):=\frac{\Gamma(c)}{\Gamma(b)} \sum_{s=0}^m  (-1)^s\frac{ \Gamma(b+s) }{ \Gamma(c+s)  } \binom{m}{s}z^s. 
\end{equation}

\begin{lem}[\textbf{Admissible sequence for the power type equilibrium measure}] \label{Lem_ak for power type eq msr}
Let $d \ge 1$ and $s \in [d-2,d)$. Define  
\begin{equation} \label{def of ak for power type eq msr}
a_k =  \frac{ \Gamma(\alpha+1+\frac{d}{2}) }{ \Gamma(\alpha+1)\Gamma(\frac{d}{2}+1) } \,  \frac{ \Gamma(-\alpha+k) }{ \Gamma(-\alpha) k! } , \qquad \alpha>-1.
\end{equation} 
(We will see in (\ref{def of power type eq msr}) that this gives rise to a well defined non-negative density.)
Then the potential \eqref{def of V gen d} associated with the corresponding density (\ref{def of eq msr mu ak gen d}) is given by 
\begin{equation} \label{def of V power}
V(x) =-\frac{d}{s} \,  \frac{ \Gamma(\alpha+1+\frac{d}{2}) }{  \Gamma(\frac{d}{2}+1) }   \frac{ \Gamma(\frac{d-s}{2})  }{ \Gamma(\alpha+1+\frac{d-s}{2}) } 
  \Big (   {}_2F_1(\tfrac{s}{2},\tfrac{s-d}{2}-\alpha;\tfrac{d}{2};|x|^2) - 1    \Big )  . 
\end{equation} 
Furthermore, if $\alpha=\frac{s-d}{2} +n$ with $n=0,1,2,\dots$, the potential $V$ in \eqref{def of V power} becomes a polynomial
\begin{align}  \label{def of V power poly}
\begin{split} 
V(x) =   \frac{ \Gamma(\frac{s}{2} +n+1) \Gamma(\frac{d-s}{2})   }{ \Gamma(\frac{s}{2}+1) }   \sum_{k=1}^n   \frac{ (-1)^{k+1}  }{k!(n-k)!} \frac{ \Gamma(k+\frac{s}{2}) }{   \Gamma(k+\frac{d}{2}) } |x|^{2k} . 
\end{split}
\end{align}
\end{lem}

As an immediate consequence of Theorem~\ref{Thm_radial} (i) and Lemma~\ref{Lem_ak for power type eq msr}, we have the following corollary.

\begin{cor}[\textbf{Power type equilibrium measure with hard edges}] \label{Cor_power type with hard edges}
Let $d \geq 1$ and $s \in [d-2,d)$, and assume $\alpha > -1$. 
Define $V_{\mathrm{h}}$ as in \eqref{def of V hard gen}, where $V$ is given by \eqref{def of V power}. Then the equilibrium measure associated with $V_{ \rm h }$ is given by
\begin{equation}  \label{def of power type eq msr}
 d\mu(x) = \frac{ \Gamma(\alpha+1+\frac{d}{2}) }{ \Gamma(\alpha+1)\Gamma(\frac{d}{2}+1) } \,  (1-|x|^2)^{\alpha}  \cdot  \mathbbm{1}_{ \{ |x| \le 1 \} } \, \frac{ \Gamma(\frac{d}{2}+1) }{\pi^{ d/2 }}\,dx.    
\end{equation}   
\end{cor}

\begin{proof}
We see that \eqref{def of power type eq msr} follows from (\ref{def of eq msr mu ak gen d}), \eqref{def of ak for power type eq msr} and the binomial expansion 
\begin{equation} \label{eq for binomial expansion}
(1-z)^{\alpha}= \sum_{k=0}^\infty \frac{ \Gamma(-\alpha+k) }{ \Gamma(-\alpha) k! } z^k, \qquad (|z|<1). 
\end{equation}
Then the corollary immediately follows from Theorem~\ref{Thm_radial} (i). 
\end{proof}

\begin{rem} We discuss some important special cases of Corollary~\ref{Cor_power type with hard edges}.
\smallskip 
\\ 
1.~(Box potential)
For the special case when $\alpha= \frac{s-d}{2}$, the potential in Corollary~\ref{Cor_power type with hard edges} reduces to the so-called \emph{box potential}, namely $V_{ \rm h }$ vanishes inside the unit ball and equals $+\infty$ outside. For this box potential, the associated equilibrium measure is given by 
\begin{equation}
d\mu_{ \rm box }(x) =  \frac{ \Gamma(\frac{s}{2}+1 )   }{ \Gamma(\frac{s-d}{2} +1)\Gamma(\frac{d}{2}+1) }  \, (1-|x|^2)^{ \frac{s-d}{2} }  \cdot  \mathbbm{1}_{ \{ |x| \le 1 \} } \, \frac{ \Gamma(\frac{d}{2}+1) }{\pi^{ d/2 }}\,dx,  
\end{equation}
as read off from (\ref{def of power type eq msr}).
This result goes back to Riesz \cite{Riesz38}; see also 
\cite[Theorem~1.1]{CSW22} and \cite{Be85}. 
We also mention that for $-2< s \le d-2$, the equilibrium measure for the box potential is known to become singular and concentrates uniformly on the boundary of the unit ball. For $s=d-1$ this result can be deduced by projecting from  a uniformly charged surface of a $d+1$-dimensional ball with respect to the Coulomb potential in dimension $d+1$; see  \cite[Proposition~4.2]{BF25a}.
\smallskip 
\\
2.~(Quadratic potential)
From (\ref{def of V power poly}), for $n=1$ the potential $V$ is quadratic, while (\ref{def of power type eq msr}) gives that the density is proportional to $(1 - |x|^2)^{1+(s-d)/2}$. It is of interest to compare this to the density profile for $d$-dimensional spinless fermions in a harmonic trap, in the scaled large $N$ limit. For the latter, one has the Thomas--Fermi density which is proportional to $(1 - |x|^2)^{d/2}$, where $|x|$ is measured in units of the maximum shell label; see e.g.~\cite{Ca06}. For $s \in [d-2,d)$ as required by Lemma \ref{Lem_ak for power type eq msr} and Corollary \ref{Cor_power type with hard edges}, the exponents of the densities only coincide for $d=1$, $s=0$. In fact then there is an equivalence at finite $N$ between the squared ground state wave function of the Fermi gas and the Boltzmann factor of the log-gas system with coupling $\beta = 2$; see e.g.~\cite{DLMS19}.

\end{rem}
  
Note that if we consider the potential $V$ in \eqref{def of V power} with soft edges, instead of $V_{\mathrm{h}}$ as in Corollary~\ref{Cor_power type with hard edges}, not every choice of $\alpha > -1$ yields a valid potential. In particular, the inequality \eqref{inequality for general} fails for arbitrary $\alpha > -1$. For instance, in \eqref{def of V power poly}, when $n$ is even the leading term of the polynomial is negative, which is not admissible for defining the model, in the sense that the second Euler-Lagrange condition cannot be satisfied. On the other hand, when $n$ is odd the potential is well defined and the associated equilibrium measure can be constructed. We formalise this in the following theorem.

\begin{thm}[\textbf{Power type equilibrium measure with soft edges}] \label{Thm_eq msr power type}
Let $d \ge 1$ and $s \in [d-2,d)$.  For $m \in \mathbb{Z}_{ \ge 0 }$, let   
\begin{align}
\begin{split}  \label{def of V poly soft}
V(x) =   \sum_{k=1}^{2m+1} b_k  |x|^{2k} ,   
\qquad b_k  =   \frac{  \Gamma(\frac{d-s}{2}) \Gamma(2m+2+\frac{s}{2} ) (-1)^{k+1}  }{\Gamma(k+\frac{d}{2})  k!(2m+1-k)!  }  \frac{ \Gamma(\frac{s}{2}+k) }{ \Gamma(\frac{s}{2}+1) } .  
\end{split}
\end{align}  
Then the associated equilibrium measure $\mu$ is given by \eqref{def of power type eq msr}, where 
\begin{equation} \label{def of alpha for power}
\alpha= \frac{s-d}{2}+2m+1. 
\end{equation}
Furthermore, the associated energy $I_V[\mu]$ is given by 
\begin{align}
\begin{split}  \label{energy for power type eq msr}
I_V[\mu]&=    \frac{ \Gamma(\frac{d-s}{2})  \Gamma(2m+2+\frac{s}{2} ) }{  \Gamma(\frac{d}{2})   }  \bigg( \frac{1}{(2m)!}\frac{1}{s} + \frac12 \sum_{k=1}^{2m+1}  \frac{ (-1)^{k+1}  }{  k!(2m+1-k)!  }  \frac{ \Gamma(\frac{s}{2}+k) \Gamma(\frac{s}{2}+2m+2) }{ \Gamma(\frac{s}{2}+1)\Gamma(\frac{s}{2}+k+2m+2) }  \bigg). 
\end{split}
\end{align}
\end{thm}

The first few polynomials in \eqref{def of V poly soft} are given as follows. 

\begin{itemize}
    \item For $m=0$, 
    \begin{equation} \label{V purely quadratic}
V(x)|_{m=0} =     \frac{  \Gamma(\frac{d-s}{2}) \Gamma(2+\frac{s}{2} ) }{  \Gamma(\frac{d}{2}+1) } |x|^2. 
\end{equation}
    \item For $m=1$, 
\begin{align}
\begin{split}
V(x)|_{m=1}&=      \frac{ \Gamma(\frac{d-s}{2}) \Gamma(4+\frac{s}{2} ) }{ \Gamma(\frac{d}{2}+1) }  \frac{   (s+2)(s+4) }{6(d+2)(d+4)  } 
\Big(  |x|^6-\frac{3(d+4)}{ s+4 }|x|^4+\frac{3(d+2)(d+4)}{ (s+2)(s+4) }|x|^2 \Big) . 
\end{split}
\end{align} 
\end{itemize}

\begin{rem}
For the potential $V$ in \eqref{def of V power}, the equality part of the Euler-Lagrange condition \eqref{EL eqn for hard wall} can be written 
\begin{equation}\label{1.34}
\int_{ |x|<1 } \frac{ (1-|x|^2)^{\alpha} }{ |x-y|^s }\,dx =  \pi^{d/2}  \frac{ \Gamma(\frac{d-s}{2}) \Gamma(\alpha+1) }{ \Gamma( \frac{d}{2} ) \Gamma(\alpha+1+\frac{d-s}{2}) }   {}_2F_1(\tfrac{s}{2},\tfrac{s-d}{2}-\alpha;\tfrac{d}{2};|y|^2), 
 \end{equation}
 where $d \ge 1$ and $|y|<1$.
This identity was also presented in \cite[Lemma~2.4]{GCO23}, based on the theory of fractional differential operators \cite[Corollary~4]{DKK17}. For the case $d=1$; see \cite[Eq.(2.3)]{GCO22} and \cite[Eq.(9)]{ADKKMMS19}. See also \cite{GP23} and the references therein for further examples of such identities.
\end{rem}

\begin{rem}[Expression in terms of the Jacobi polynomial]
Recall the hypergeometric expression of the Jacobi polynomial \cite[Eq.(15.9.1)]{NIST}
\begin{equation}
P_n^{ (a,b) }(x) = \frac{ \Gamma(n+a+1) }{ \Gamma(a+1) \,n! } {}_2F_1( -n,n+a+b+1; a+1; \tfrac{1-x}{2} ). 
\end{equation} 
By using this, one can observe that the potential $V$ in \eqref{def of V power poly} can be expressed as   
\begin{equation}
V(x)= - \frac{2}{s}   \frac{   \Gamma(\frac{d-s}{2})  \Gamma(2m+2+\frac{s}{2} ) }{  \Gamma(2m+1+\frac{d}{2}) }    \bigg( P_{ 2m+1 }^{ (a,b) }(1-2|x|^2) - \frac{ \Gamma(2m+1+\frac{d}{2}) }{ (2m+1)! \Gamma( \frac{d}{2} ) }  \bigg),  
\end{equation}
where 
\begin{equation}
 a= \frac{d}{2}-1, \qquad b=-2m-1+\frac{s-d}{2}.
\end{equation}  
\end{rem}

\begin{rem}[Consistent check with \cite{ADKKMMS19} for the one dimensional Riesz gas]
For the case $d=1$ and $m=0$, the purely quadratic potential \eqref{V purely quadratic} takes the form 
\begin{align}
\begin{split} \label{V purely quadratic d=1}
V(x) =    \frac{ 2\, \Gamma( \frac{1-s}{2}) \Gamma( 2+\frac{s}{2}) }{ \sqrt{\pi}  } \, x^2 . 
\end{split}
\end{align} 
The associated limiting measure is then given by \eqref{def of power type eq msr}, supported on $[-1,1]$, with exponent $\alpha=\tfrac{s+1}{2}$.

We now compare this with the results of \cite{ADKKMMS19}, where the one-dimensional Riesz gas with Hamiltonian
\begin{equation}
\frac12 \sum_{j=1}^N x_j^2 + \frac{J \textup{sgn}(s)}{ 2} \sum_{ j\not=k } \frac{1}{|x_j-x_k|^s} =  \frac{J \textup{sgn}(s) s} {2}\bigg(   \frac{1}{J \textup{sgn}(s) s} \sum_{j=1}^N x_j^2 + \frac{1}{s} \sum_{ j\not=k } \frac{1}{|x_j-x_k|^s}  \bigg)
\end{equation}
was studied. Let $B(x,y)$ denote the gamma function evaluation of the Euler beta integral. In the large-$N$ limit, the limiting density was shown to be
\begin{equation}
\ell_s^{-1} F_s( y/\ell_s ), \qquad F_s(z)= \frac{1}{ B( \gamma_s+1, \gamma_s+1 ) } \Big( \frac14-z^2\Big)^{ \gamma_s }.
\end{equation}
Here, for $-2<s<1$, 
\begin{equation}
\gamma_s=  \frac{s+1}{2}, \qquad 
\ell_s = \Big( \frac{ J |s| \pi (s+1) }{ \sin( \frac{\pi}{2}(s+1) )  B(\frac{s+3}{2},\frac{s+3}{2}) } \Big)^{ \frac{1}{s+2} }. 
\end{equation} 
Note that $\gamma_s$ coincides with $\alpha=\tfrac{s+1}{2}$.  
Furthermore, to enforce $\ell_s=2$ so that the limiting support is $[-1,1]$, we choose
\begin{equation}
J= \frac{ \sin( \frac{\pi}{2}(s+1) )  B(\frac{s+3}{2},\frac{s+3}{2}) }{  |s| \pi (s+1) } 2^{s+2}.
\end{equation}
Then since
\begin{align*}
\frac{1}{ J \textup{sgn}(s) s}= \frac{\pi}{  \sin( \frac{s+1}{2}\pi )  } \frac{  |s|   (s+1) }{  B(\frac{s+3}{2},\frac{s+3}{2}) 2^{s+2} \textup{sgn}(s) s }  =   \frac{\pi   }{ \sin( \frac{1-s}{2} \pi ) } \frac{ 2\, \Gamma( 2+\frac{s}{2}) }{ \sqrt{\pi} \Gamma( \frac{1+s}{2}) },
\end{align*}
one can see that the resulting quadratic potential is consistent with \eqref{V purely quadratic d=1}.
\end{rem}

\subsection{Power type external potential}\label{S1.3}

As a second prominent example, we consider the purely power type external potential. As before, we begin with introducing the associated sequence.

\begin{lem}[\textbf{Admissible sequence for the power type external potential}] \label{Lem_ak for power type potential}
Let $d \ge 1$ and $s \in [d-2,d)$. For $p\in \mathbb{Z}_{ \ge 1 }$, let 
\begin{equation} \label{def of ak for power type potential}
a_k= - \frac{ \sin(\pi \frac{d-s}{2})  }{\pi} \frac{\Gamma(1+\frac{s}{2})}{ \Gamma( \frac{d}{2} ) } \frac{2p+s}{d}\, \frac{ \Gamma(k+\frac{d-s}{2}) }{ (k+\frac{d-s}{2}-p)k! } . 
\end{equation}
(We will see in (\ref{def of power type potential eq msr}) that this gives rise to a well defined non-negative density.)
Then the potential \eqref{def of V gen d} associated with this sequence is given by \begin{equation} \label{def of V for power type}
V(x) =    \frac{ \Gamma(\frac{d-s}{2})\Gamma(\frac{s}{2}+p)   }{ \Gamma(\frac{d}{2}+p)  } \, \frac{2p+s}{ 2p  }       |x|^{2p} . 
\end{equation}
\end{lem}

The purely power type potential \eqref{def of V for power type} has been extensively investigated for log gases. For $d=1$, it is commonly referred to as the \emph{Freud potential} (see e.g. \cite{CKM23}), while for $d=2$ it is known as the \emph{Mittag--Leffler potential} (see e.g. \cite{Ch22,Ch23,ACCL24,Be25}).
However, for general Riesz gases, the purely power-type potential has not been extensively investigated, and we formulate its equilibrium measure in the following theorem.

\begin{thm}[\textbf{Power type external potential with soft edges}] \label{Thm_eq msr for power potential} Let $d \ge 1$ and $s \in [d-2,d)$. Let  $p\in \mathbb{Z}_{ \ge 1 }$. Then the equilibrium measure $\mu$ associated with the purely power type potential $V$ in \eqref{def of V for power type} is given by 
\begin{equation} \label{def of power type potential eq msr}
d\mu(x) =  \frac{\Gamma(1+\tfrac{s}{2})  } { \Gamma(1+\frac{s-d}{2}) \Gamma(\frac{d}{2}+1) }   \frac{ 2p+s }{ 2p-d+s } \, {}_2F_1( \tfrac{d-s}{2}, \tfrac{d-s}{2}-p; \tfrac{d-s}{2}+1-p; |x|^2 ) \cdot  \mathbbm{1}_{ \{ |x| \le 1 \} } \,  \frac{ \Gamma(\frac{d}{2}+1) }{\pi^{ d/2 }}\,dx   .
\end{equation} 
Furthermore, we have 
\begin{equation} \label{energy for power type potential}
I_V[\mu]=  \frac{ \Gamma( \tfrac{d-s}{2} ) \Gamma(\frac{s}{2})}{ \Gamma( \frac{d}{2} ) } \, \frac{(2p+s)^2}{ 2 p(4p+s)  }  . 
\end{equation}
\end{thm}

See Figure~\ref{Fig_density of power type potential} for the graphs of $\mu(|x|)$. 

\begin{rem}[Purely quadratic case]
We consider the case $p=1$. Then \eqref{def of V for power type} is given by 
\begin{equation}
V(x)|_{p=1} =  \frac{ \Gamma(\frac{d-s}{2})  \Gamma(\frac{s}{2}+2)   }{ \Gamma(\frac{d}{2}+1)  } \,       |x|^{2}. 
\end{equation}
This coincides with the special case \eqref{V purely quadratic} of the previous example. Note that in the one-dimensional case $d=1$, this purely quadratic potential has been extensively studied in \cite{ADKKMMS19}, where all regimes with $s > -2$ have been analysed, including both the long-range $(s<1)$ and short-range $(s>1)$ interaction regimes. 
We also note that for the log gas case, the energy \eqref{energy for power type potential} becomes 
\begin{align}
\lim_{s\to0}\Big(I_V[\mu]|_{p=1}-\frac{1}{s}\Big)= \begin{cases}
\frac{3}{4}+\log 2  &\textup{if } d=1,
\smallskip 
\\
\frac34 &\textup{if } d=2. 
\end{cases} 
\end{align}
This corresponds to the weighted energy of the semi-circle and circular law, respectively. 

Although Theorem~\ref{Thm_eq msr for power potential} is, in principle, valid for $s \ge d-2$, the explicit formula for the equilibrium measure in \eqref{def of power type potential eq msr} remains well defined for $p=1$ even in the extended regime $s = d-3$. 
In this case one finds
\begin{equation}
 \label{sdm3}
d\mu(x) = \frac{\Gamma( \frac{1+d}{2} ) }{\pi^{(1+d)/2}} \frac{1}{\sqrt{1-x^2}} \,dx \;,
\end{equation} 
which coincides exactly with the result obtained in \cite{CSW22}. If one sets $s=d-3$ and consider the case $p>1$ we observe the following phenomenon. For instance if one sets $p=2$, the formula in (\ref{def of power type potential eq msr}), with $s=d-3$, yields
\begin{equation}
 \label{sdm3_p2}
d\mu(x) = \frac{1}{\pi^{(1+d)/2}} \frac{d+1}{2} \Gamma\Big( \frac{d-1}{2}\Big) \frac{2x^2-1}{\sqrt{1-x^2}} \,dx \;.
\end{equation} 
Clearly, this density is nonnegative only for $|x|\in[\sqrt{2}/2,1]$.
Although the expression in \eqref{sdm3_p2} is likely not the correct equilibrium profile---since, in particular, it is not properly normalised over $|x|\in[\sqrt{2}/2,1]$---its qualitative behaviour is reminiscent of the observation made in \cite[Section~2.1.3]{CSW23}. There, numerical simulations in the case $s=d-3$ and $p>1$ suggest that the equilibrium measure is not supported on the entire unit ball, but rather concentrates on a spherical shell. 
\end{rem}

\begin{rem}[Freud potential for the one-dimensional log gas $d=1, s=0$]
For $d=1$ and $s=0$, the potential \eqref{def of V for power type} takes the form  
\begin{equation}
V(x) |_{d=1,s=0}= \frac{ \sqrt{\pi}\, \Gamma( p ) }{ \Gamma(p+\frac12) }  x^{2p}.
\end{equation}
Then the associated equilibrium measure is supported on $[-1,1]$ with density
\begin{equation}
\frac{ 2p }{ \pi  } |x|^{2p-1} \int_1^{ 1/|x| } \frac{ u^{2p-1} }{ \sqrt{u^2-1} }\,du ; 
\end{equation}
see e.g. \cite[Chapter IV.5]{ST97} and \cite[Eq.(1.9)]{CKM23}. 
For an application to random matrix theory, specifically to products of random matrices, we refer the reader to \cite{Fo14}.
Here, compared to \cite[Eq.(1.9)]{CKM23}, we take a proper normalisation so that the support becomes $[-1,1]$.
This density can be expressed in terms of the incomplete beta function as  
\begin{equation} 
-\frac{p}{\pi} |x|^{2p-1} B_{x^2}(\tfrac12-p,\tfrac12) = \frac{1}{\pi} \frac{2p}{2p-1} {}_2F_1( \tfrac12,\tfrac12-p; \tfrac32-p;x^2 ), 
\end{equation}
where we have used \cite[Eq.(8.17.7)]{NIST}. Therefore, one can notice that this density is a special case of \eqref{def of power type potential eq msr}.  
See also \cite{Fo14} for an application of Freud weight in the context of product of random matrices. 
\end{rem}

\begin{rem}[Edge behaviour]
By applying the Gauss summation formula \eqref{def of Gauss sum}, one can observe that for \( d-2 < s < d \), the density \eqref{def of power type potential eq msr} vanishes at the edge \( |x| = 1 \). 
Moreover, using the differentiation formula for the hypergeometric function 
(see, for example, \cite[Eq.(15.5.1)]{NIST}), we obtain
\begin{equation}
\mu(|x|)=O\Big( (1-|x|^2)^{ 1- \frac{d-s}{2} } \Big), \qquad \textup{as }|x| \to 1. 
\end{equation}
In the special case \( d=1 \) and \( s=0 \), this reproduces the classical 
square-root decay at the boundary for one-dimensional log gases. 
On the other hand, for the Coulomb case \( s = d-2 \), the density exhibits 
a discontinuity at the edge (the circular law provides a well-known example of this situation).
\end{rem}

\begin{figure}[t]
	\begin{subfigure}{0.32\textwidth}
		\begin{center}	
			\includegraphics[width=\textwidth]{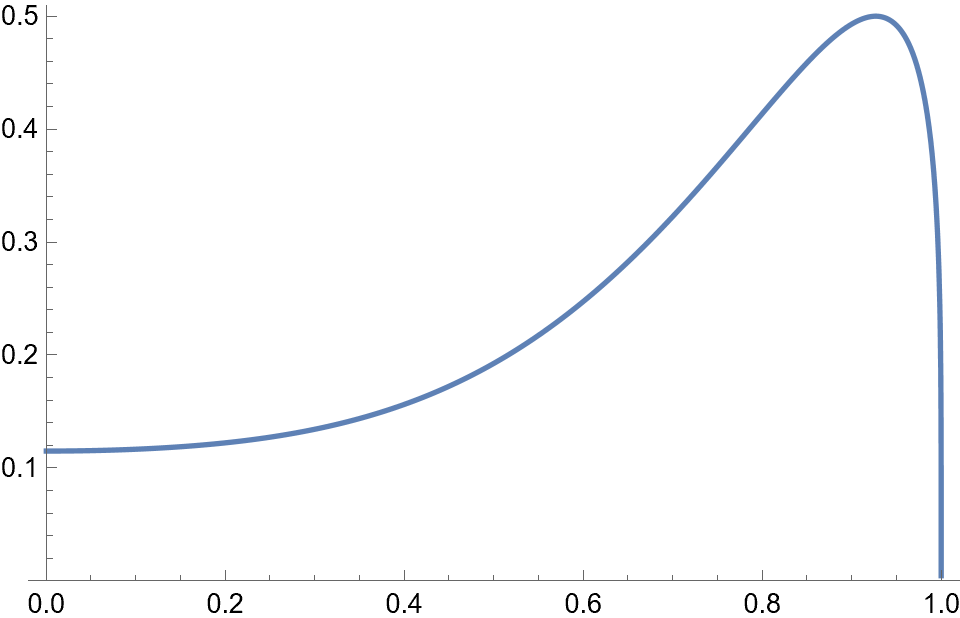}
		\end{center}
		\subcaption{$s=1/2$}
	\end{subfigure}	
		\begin{subfigure}{0.32\textwidth}
		\begin{center}	
			\includegraphics[width=\textwidth]{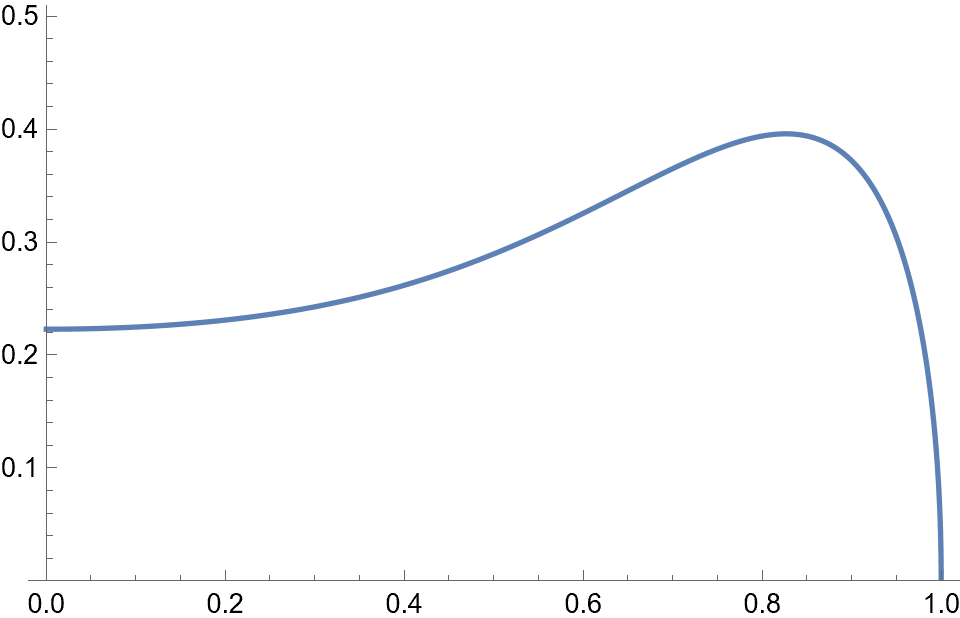}
		\end{center}
		\subcaption{$s=1$}
	\end{subfigure}
    	\begin{subfigure}{0.32\textwidth}
		\begin{center}	
			\includegraphics[width=\textwidth]{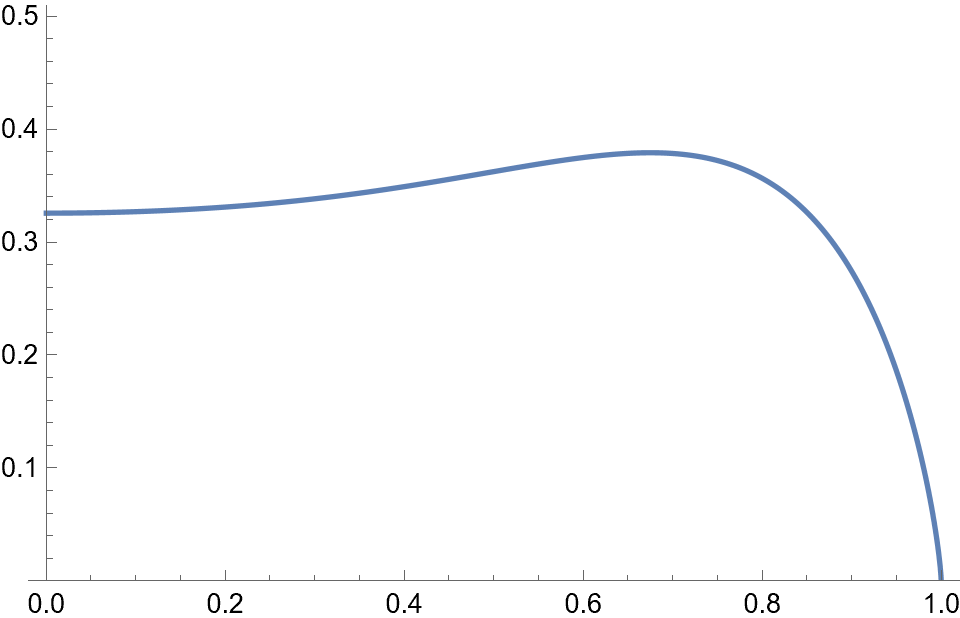}
		\end{center}
		\subcaption{$s=3/2$}
	\end{subfigure}	
	\caption{Graphs of the density $|x| \to \mu(|x|)$ in \eqref{def of power type potential eq msr}, where $d=2$ and $p=3$.} \label{Fig_density of power type potential} 
\end{figure}

\subsection{A constrained half plane equilibrium problem in the Coulomb case}\label{S1.4} 
In this subsection, we present another application of Theorem~\ref{Thm_radial} for Coulomb gases with a half-space constraint.

We begin by discussing the two-dimensional case.
Let $x = (x_0,x_1) \in \mathbb R^2$. For the two-dimensional Coulomb potential $- \log |x|$, consider the problem of finding the equilibrium measure when the external potential is given by
\begin{equation}\label{V1}
V({x}) = \begin{cases}
|{x}|^2, &\textup{if } x_0 > a,
\smallskip 
\\
+\infty, &\textup{otherwise}. 
\end{cases}
\end{equation}
In words, there is a hard wall to the left of $x_0 = a$, while to the right there is a confining harmonic potential. For $a=-\infty$ the former has no effect, and we know from 
(\ref{def of eq msr mu ak gen d}) and (\ref{e.3}) with $a_0=1$, $a_1=a_2=\cdots=0$ and $d=2$  that the equilibrium measure is the constant density ${1 \over \pi}$ supported on the unit disk. Moreover, the fact that the case $a=-\infty$ corresponds to soft wall boundary conditions tells us that this remains the equilibrium measure for all $a <-1$. In \cite{ASZ14} it was established that for $a \in (-1,\sqrt{2})$ that the equilibrium measure decomposes as a singular component supported on $x_0 = a$ and an absolutely continuous  component supported in $\mathbb R^2$. Moreover, for $a \ge \sqrt{2}$, the equilibrium density is fully supported on $x_0=a$
with the semi-circle functional form
\begin{equation}\label{V2}
\sigma(x) = \delta(x_0-a) {\sqrt{2 - x_1^2} \over \pi}
\mathbbm 1_{\{|x_1| \le \sqrt{2}\}}.
\end{equation}
 
We now consider the $(d+1)$-dimensional Riesz gas with parameter $s$, specialised to the Coulomb case so that
$s=d-1$.   
Let $\hat{x}=(x_0,x_1,\dots,x_d)$, 
and with $a \in \R$ consider the potential
\begin{equation}\label{3.1}
V(\hat{x}) = \begin{cases}
|\hat{x}|^2, &\textup{if } x_0 > a,
\smallskip 
\\
+\infty, &\textup{otherwise}. 
\end{cases}
\end{equation} 
Note that if $a=-\infty$, thus there is no constraint, then the equilibrium measure is uniform on the unit sphere in $\R^{d+1}$, 
 \begin{equation}  \label{def of eq msr without constriant}
d\mu_{V}(\hat{x})|_{a=-\infty} = \frac{ \Gamma(\frac{d+3}{2}) }{\pi^{ (d+1)/2 }}\, \mathbbm{1}_{ \{ |\hat{x}| \le 1 \} } \,d\hat{x}; 
\end{equation} 
this follows from the case $a_0=1$, $a_1=a_2=\cdots=0$ of (\ref{def of eq msr mu ak gen d}) and (\ref{e.3}). 

For \( d = 0 \) and \( d = 1 \), an interesting critical phenomenon was obtained in \cite{DKMSS17, ASZ14}. As described in the case $d=1$ in Section \ref{S1.4}, this is characterised by a critical value \( a_{\mathrm{cri}} = a_{\mathrm{cri}}(d) \) separating two of three possible regimes. We revisit those here. 
\begin{itemize}
    \item[(i)] \textbf{No effective hard wall}: \( a < -1 \). In this regime, the hard-wall constraint has no effective influence; consequently, the equilibrium measure coincides with that in \eqref{def of eq msr without constriant} for the unconstrained case.
    \smallskip
    \item[(ii)] \textbf{Partially effective hard wall}: \( a \in (-1, a_{\mathrm{cri}}) \). The equilibrium measure consists of an absolutely continuous component supported in \( \mathbb{R}^{d+1} \), together with a singular component supported on the codimension-one set \( \{a\} \times \mathbb{R}^d \).
    \smallskip
    \item[(iii)] \textbf{Fully effective hard wall}: \( a > a_{\mathrm{cri}} \). The equilibrium measure becomes completely confined to \( \{a\} \times \mathbb{R}^d \).
\end{itemize}

To be more precise the following was shown.
\begin{itemize}
    \item In one-dimension $d+1=1$, it was shown in \cite{DKMSS17} that the threshold is at $a_{ \rm cri }(0)=1$, and for $a>1$, the measure is purely Dirac measure 
    \begin{equation}
    d\mu_V(x)= \delta_a(dx_0). 
    \end{equation} 
    \item In two-dimension $d+1=2$, it was shown in \cite{ASZ14} that the threshold is at $a_{ \rm cri }(1)=\sqrt{2}$. For $a>\sqrt{2}$, the measure is purely one-dimensional support 
    with corresponding density (\ref{V2}). (See also Remark~\ref{Rem_elliptic half space} for an extension to the elliptic potential interpolating between the one- and two-dimensional regimes.)
\end{itemize} 
Such a phenomenon is a special feature of the Coulomb case. For the general Riesz gas, it might be that no fully effective hard wall arises. For instance, the log gas in one-dimension, the resulting measure is absolutely continuous with respect to the one-dimensional measure. Many works have addressed this situation; as an incomplete listing see \cite{DM06,DM08,MNSV09,BDG01,KC10,RKC12,MV09,FW12,MS14}.
One should note for the Coulomb case,  the effect of the half plane constraint is not that of a balayage measure \cite{AR17, Ad18, Ch23a} (see also the recent review \cite[Section 5]{BF25a}), where all the displaced charge density due to the barrier constraint accumulates on the boundary.
The latter would hold,
 for example, if the condition $x_0 > a$ in (\ref{3.1}) was changed to $|\hat{x}| > a$. 

We can establish a partial analogue for higher dimensions.

 \begin{thm}[\textbf{Coulomb gases constrained on the half space}] \label{Thm_Coulomb half space}
Let $d=0,1,2,\dots$ and define 
\begin{equation} \label{acrit}
a_{\rm cri} \equiv a_{\rm cri}(d):= (d+1) \Big( \frac{ \Gamma(\frac{d}{2}+1)  }{ \sqrt{\pi} \, \Gamma(\frac{d+3}{2})   } \Big)^{ \frac{d}{d+1} }.
\end{equation}
Then we have the following.
\begin{itemize}
    \item[\textup{(i)}] The equilibrium measure is supported in $\R^d$ only if $a \ge a_{ \rm cri }$.  
    \smallskip 
     \item[\textup{(ii)}] Suppose that the equilibrium measure is supported in $\R^d$. 
    Introduce the notation $\hat{x}=(x_0,x) \in \R^{d+1}$ with $x=(x_1,\dots,x_d) \in \R^d$. Then the equilibrium measure is given by 
      \begin{equation} 
   d\mu_V(\hat{x})= \delta_{ a }(dx_0)\, d\mu_W(x), 
    \end{equation} 
    where $\hat{x}=(x_0,x) \in \R^d$ with $x=(x_1,\dots,x_d) \in \R^d$, and 
\begin{equation} \label{def of eq msr in a Rd}
d\mu_W(x)= \frac{2R}{\pi}   \,  (1-|x/R|^2)^{ \frac12 }  \cdot  \mathbbm{1}_{ \{ |x| \le R \} } \, \frac{ \Gamma(\frac{d}{2}+1) }{\pi^{ d/2 }} \,dx,  \qquad    R= \Big( \frac{ \sqrt{\pi} \, \Gamma(\frac{d+3}{2})   }{ \Gamma(\frac{d}{2}+1)  } \Big)^{ \frac{1}{d+1} }.  
\end{equation}
Furthermore, for any $\beta>0$, we have 
\begin{equation} \label{LD Prob in main}
\lim_{N \to \infty} \frac{ \log   \mathbb{P}[x_0 \ge a] }{ N^{ \frac{d+3}{d+1} } } = -\frac{\beta}{2} C(a;d),
\end{equation}
where  
\begin{equation}\label{Cad}
C(a;d)= a^2+ \Big( \frac{ \sqrt{\pi} \, \Gamma(\frac{d+3}{2})   }{ \Gamma(\frac{d}{2}+1)  } \Big)^{ \frac{2}{d+1} }\frac{d(d+1)}{(d-1)(d+3)}-  \frac{(d+1)^2}{(d-1)(d+3)} . 
\end{equation} 
\end{itemize} 
\end{thm}

The values of $C(a,d)$ for the first few values of $d$ are 
\begin{align}
C(a;0)=a^2+\frac{1}{3},
\qquad
 C(a;1) =a^2+ \frac{\log 2}{2}, \qquad
  C(a;2) =  a^2  +\frac{5}{9} \Big( \frac{(6\pi)^{3/2}}{6}-1 \Big). 
\end{align} 
Note that in (\ref{Cad}), there is a removable singularity at $d=1$, since $\lim_{d\to 1} C(a,d) = C(a,1)$ exists. Also, it is easy to check that $C(a;d)>a^2$. 
The leading asymptotic behaviour of $\mathbb{P}[ x_0 >a  ]$ in (\ref{LD Prob in main}) is similar in spirit to many works using the so-called Coulomb gas method; see e.g. \cite{MNSV11} for a result of this type relating to the case $d=1$ and $a=0$, and the review \cite{Fo14a} for an historical introduction and development.
 
As a conjecture, we propose the following, as is known from
\cite{ASZ14,DKMSS17} for $d=0,1$.

\begin{conj}\label{Cj1}
The equilibrium measure is supported in $\R^d$ if and only if $a \ge a_{ \rm cri }$.  
\end{conj}

At present, we are not aware of a rigorous proof of this statement in full
generality. Nevertheless, in Appendix~\ref{Appendix_d=3} we establish
Conjecture~\ref{Cj1} in the special case \( d=3 \).

\begin{rem} \label{Rem_conjecture inequality}
By Lemma~\ref{Lem_integral ev for half space constraint}, the conjecture reduces to proving that, for $a \ge a_{\mathrm{cri}}$, $F(t,x)\ge F(0,0)$ for any $t \ge 0$ and $x \in \R^d$, where $F(t,x)$ is given by \eqref{def of F(t,x)}. 
Moreover, Lemma~\ref{Lem_EL for the vertical} shows that $F(t,0) \ge F(0,0)$ for all $t \ge 0$. Hence, it suffices to establish that $F(t,x) \ge F(t,0)$.  
To be more concrete, for given $t \ge 0$, we introduce
\begin{equation} \label{def of G(x)}
G(x):= \mathcal{G}(x) + x^2 , 
\end{equation} 
where 
\begin{equation}
\mathcal{G}(x):= \frac{4d}{d-1} \frac{1}{\pi}  \int_0^1 \Big( t^2+(x+r)^2  \Big)^{ -\frac{d-1}{2} } {}_2F_1\Big( \frac{d-1}{2}, \frac{d-1}{2}; d-1;  \frac{ 4 x r }{ t^2+( x+r )^2 }\Big) (1-r^2)^{\frac12} r^{d-1}\,dr. 
\end{equation}
With this notation, the conjecture follows once we show that $G(x) \ge G(0)$ for any $x \ge 0$. 

This property, while seemingly simple, appears to be rather delicate to verify directly.
In fact, one readily observes that $\mathcal{G}(x)-\mathcal{G}(0) \le 0$ for $x \ge 0$. Consequently, in order to establish the non-negativity of $G(x)-G(0)$, it is necessary to show that the quadratic term $x^2$ compensates for this negative contribution. Figure~\ref{Fig_graph of increasing G} illustrates the behaviour of the function $x \mapsto G(x)-G(0)$ for several values of $d$ and $t$.

We further remark that, as a consequence of
Theorem~\ref{Thm_eq msr for power potential} (or equivalently
Theorem~\ref{Thm_eq msr power type}) applied to the quadratic potential $|x|^2$, the case $t=0$ can be treated explicitly.
More precisely, when $t=0$ one has $G(x)=0$ if $x\le 1$ and $G(x) \ge 0$ if $x \ge 1$; see Figure~\ref{Fig_graph of increasing G} (A).  

Let us also note that, in the special case \( d=3 \), the function \( G(x) \)
admits an explicit representation,
\begin{equation}
G(x)|_{d=3} = 3t^2+\frac32 + \frac{1}{x}  \re \Big[(x+it)^2-1)^{\frac32}\Big] ; 
\end{equation}
see \eqref{G(x) for d=3}. Using this, we show in
Appendix~\ref{Appendix_d=3} that \( G(x) \) is an increasing function of \( x \) when \( d=3 \).
\end{rem}

\begin{figure}[t]
	\begin{subfigure}{0.32\textwidth}
		\begin{center}	
			\includegraphics[width=\textwidth]{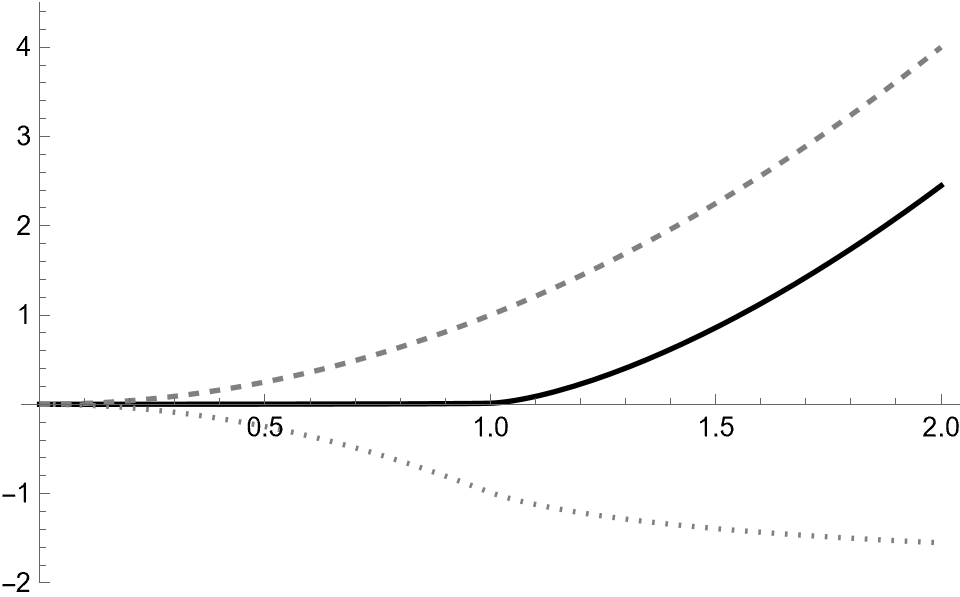}
		\end{center}
		\subcaption{$d=2,t=0$}
	\end{subfigure}	
		\begin{subfigure}{0.32\textwidth}
		\begin{center}	
			\includegraphics[width=\textwidth]{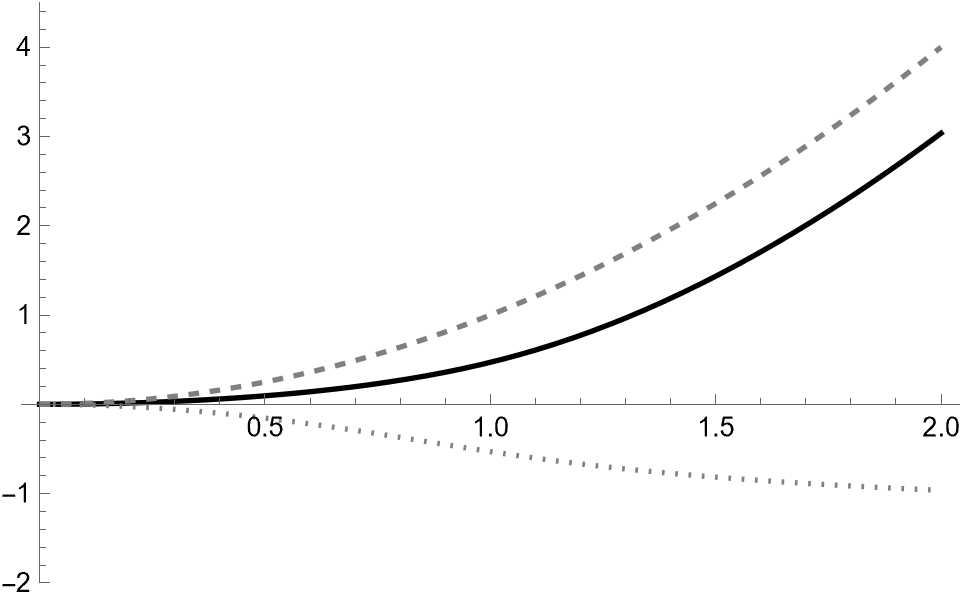}
		\end{center}
		\subcaption{$d=2,t=0.3$}
	\end{subfigure}
    	\begin{subfigure}{0.32\textwidth}
		\begin{center}	
			\includegraphics[width=\textwidth]{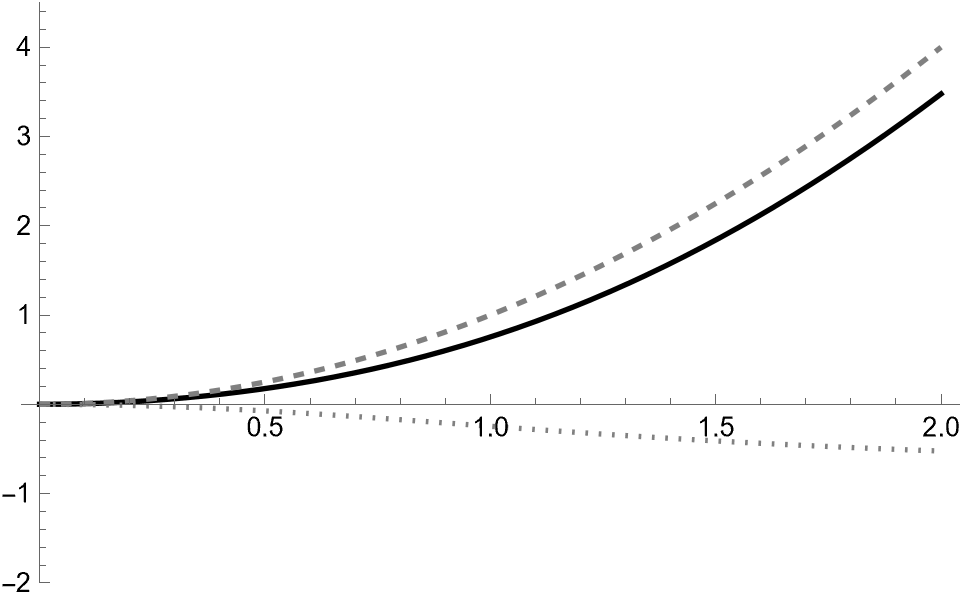}
		\end{center}
		\subcaption{$d=2,t=0.7$}
	\end{subfigure}	
	\caption{Graphs of the density $x \mapsto G(x)-G(0)$.}  \label{Fig_graph of increasing G}
\end{figure}

\begin{rem}[The limit of large dimension $d$] In this limit one finds from (\ref{acrit}) that 
\begin{equation}
a_{\rm cri}(d) = \sqrt{\frac{2d}{\pi}} + \frac{1}{\sqrt{2\pi}} \frac{\log(d)}{\sqrt{{d}}}+O(d^{-1/2}), 
\end{equation}
while from (\ref{Cad}) one has 
\begin{equation}
C(a,d) = a^2 + \frac{\log(d)}{d} + O(d^{-1}). 
\end{equation} 
\end{rem}

\medskip

\subsection*{Acknowledgments} SSB was supported by the National Research Foundation of Korea grants (RS-2023-00301976, RS-2025-00516909).
Funding support to PJF for this research was through the Australian Research Council Discovery Project grant DP250102552. SNM and GS acknowledge support from ANR Grant No. ANR-23-CE30-0020-01 EDIPS. 
This collaboration was begun during the MATRIX program “Log-gases in Caeli Australi”, held in Creswick, Victoria, Australia during the first half of August 2025. We thank the organisers for facilitating this, and for the stimulation that the program provided.

\section{Proofs of Theorems~\ref{Thm_radial}, \ref{Thm_eq msr power type}, and \ref{Thm_eq msr for power potential}}\label{S2}

This section is devoted to the proofs of Theorems~\ref{Thm_radial}, \ref{Thm_eq msr power type}, and \ref{Thm_eq msr for power potential}. Each proof is presented in a separate subsection below.

\subsection{Proof of Theorem~\ref{Thm_radial}}

In this subsection, we prove Theorem~\ref{Thm_radial}.  It is sufficient to consider the case $d-2<s \le d$, since the Coulomb gas case $s=d-2$ follows from consideration of the Poisson equation, as previously discussed.

For the probability measure $\mu$ given by \eqref{def of eq msr mu ak gen d}, let $f$ be the radial density, i.e.  
\begin{equation} \label{def of radial density ak}
f(r) =  d \sum_{k=0}^\infty a_k r^{2k+d-1}. 
\end{equation}
Notice here that by \eqref{def of norm admissible an}, we have $\int_0^1 f(r)\,dr=1$. 

The overall strategy for the proof of Theorem~\ref{Thm_radial} is summarised as follows.

\begin{itemize}
    \item In Lemma~\ref{Lem_angular integral gen}, using rotation invariance, we rewrite the potential 
    $
    \int |x-y|^{-s}\, d\mu(x)
    $
    originally given as a $d$-dimensional integral, in terms of a one-dimensional real integral involving the radial density \eqref{def of radial density ak}. Here, the key input is the Funk-Hecke formula \eqref{eq for Funk-Hecke}. As a consequence, the expression contains a certain hypergeometric function; see \eqref{2.2}. 
    \smallskip 
   \item In Lemma~\ref{Lem_integral expansion gen}, by applying the Gauss series expansion of the hypergeometric function \eqref{def of 2F1 Gauss series} together with the quadratic transformation \eqref{quad trans 2F1}, we express the integral from the previous step as a series involving the moments of the radial density. Up to this step, we exclude the case $d=1$ and assume $d \ge 2$. 
\smallskip 

\item In Proposition~\ref{Prop_evaulation of the EL conditions}, we evaluate this series using the expansion \eqref{def of radial density ak}. A crucial step here is an algebraic identity (Lemma~\ref{Lem_algebraic equal}) which allows us to significantly simplify the resulting expression. This proposition also applies when $d=1$, in which case the result follows directly without the previous reduction steps. Consequently, for general $d \ge 1$ and $s \in (d-2,d)$, we can explicitly rewrite the Euler--Lagrange equations in the form stated in Theorem~\ref{Thm_radial}.
\end{itemize}
We begin with the first step.  

\begin{lem}  \label{Lem_angular integral gen}
Let $d \ge 2$ and $s \in (d-2,d)$. For the probability measure $\mu$ in \eqref{def of eq msr mu ak gen d}, let $f$ be the associated radial density \eqref{def of radial density ak}. 
Let $y \in \mathbb{R}^{d}$. Then we have 
\begin{align} \label{2.2}
\int \frac{ d\mu(x) }{ |x-y|^s } =  \int_{ 0 }^1    \frac{1}{(|y|+r)^s}    \, {}_2F_1\Big(  \frac{s}{2}, \frac{d-1}{2} \, ; \, d-1 \, ; \,   \frac{4r|y|}{(|y|+r)^2} \Big)\, f(r) \,dr. 
\end{align}
Consequently, the Euler-Lagrange variational conditions \eqref{def of EL eqn} are given by 
\begin{equation}
\frac{2}{s} \int_{ 0 }^1    \frac{1}{(|y|+r)^s}    \, {}_2F_1\Big(  \frac{s}{2}, \frac{d-1}{2} \, ; \, d-1 \, ; \,   \frac{4r|y|}{(|y|+r)^2} \Big)\, f(r) \,dr +V(y)  \begin{cases}
= c & \textup{if } |y| \le 1,
\smallskip 
\\
\ge c & \textup{if } |y| > 1. 
\end{cases}
\end{equation}
\end{lem}

\begin{proof}
First, we recall the Funk-Hecke formula; see e.g. \cite[Appendix A.2]{CSW22}. For $d \ge 2$, let $\sigma_{S_1}$ be the uniform probability measure on the unit centered sphere $S_1= \{ x\in\R^d: |x|=1 \}. $ Then for $ z \in \R^d $ with $|z|=1$, we have 
\begin{equation} \label{eq for Funk-Hecke}
\int_{ S_1 } g(z \cdot x) \,d\sigma_{S_1}(x) = \frac{ \Gamma( \frac{d}{2} ) }{ \sqrt{\pi}\, \Gamma( \frac{d-1}{2} ) } \int_{-1}^1 g(t) (1-t^2)^{ \frac{d-3}{2} }\,dt. 
\end{equation}
By applying this, we have 
\begin{align}
\begin{split} \label{2.8}
\int \frac{ d\mu(x) }{ |x-y|^s } & =\frac{ \Gamma(\frac{d}{2}+1) }{\pi^{ d/2 }} \frac{1}{d} \int_0^1 \bigg( \int_{ S_1 } \Big( |y|^2+r^2-2|y|r \frac{y}{|y|} \cdot u \Big)^{ -\frac{s}{2} } \,du \bigg) f(r)\,dr
\\
&=  \frac{ 1  }{ \sqrt{\pi} }  \frac{  \Gamma(\frac{d}{2} ) }{   \Gamma( \frac{d-1}{2} ) }    \int_0^1 \bigg( \int_{-1}^1   (1-t^2)^{ \frac{d-3}{2} } \Big( |y|^2+r^2-2|y|r t \Big)^{ -\frac{s}{2} }   \,dt \bigg) f(r)\,dr.   
\end{split}
\end{align}
Here, we have used $|S_1|= 2\pi^{d/2}/\Gamma(d/2)$. 
Note that 
\begin{align*}
&\quad \int_{-1}^1   (1-t^2)^{ \frac{d-3}{2} } \Big( |y|^2+r^2-2|y|r t \Big)^{ -\frac{s}{2} }   \,dt  = \int_0^{\pi}  \sin^{ d-2 }(\theta) \Big( |y|^2+r^2-2|y|r \cos(\theta) \Big)^{ -\frac{s}{2} }   \,d\theta 
 \\
 &= 2^{d-2}\int_0^{\pi} \sin^{ d-2 }(\theta/2)\cos^{d-2}(\theta/2) \Big( (|y|+r)^2-4|y|r \cos^2(\theta/2) \Big)^{ -\frac{s}{2} }    \,d\theta 
 \\
 &= \frac{ 2^{d-1} }{ (|y|+r)^s } \int_0^{\pi/2}  \sin^{ d-2 }(\theta)\cos^{d-2}(\theta) \Big( 1-\frac{4|y|r}{ (|y|+r)^2 } \cos^2(\theta) \Big)^{ -\frac{s}{2} }    \,d\theta . 
\end{align*}
Since 
$$
\frac{4|y|r}{(|y|+r)^2}<1, \qquad \textup{for } r \in (0,1), 
$$
it follows from the binomial expansion \eqref{eq for binomial expansion} that
\begin{align*}
&\quad  \int_0^{\pi/2}  \sin^{ d-2 }(\theta)\cos^{d-2}(\theta) \Big( 1-x \cos^2(\theta) \Big)^{ -\frac{s}{2} }    \,d\theta  =  \sum_{k=0}^\infty \frac{ \Gamma(\frac{s}{2}+k) }{ \Gamma(\frac{s}{2}) k! }x^k   \int_0^{\pi/2}  \sin^{d-2}(\theta)\cos^{2k+d-2}(\theta) \, d\theta   
 \\
 &= \frac12 \sum_{k=0}^\infty \frac{ \Gamma(\frac{s}{2}+k) }{ \Gamma(\frac{s}{2}) k! }  \frac{ \Gamma( \frac{d-1}{2}  ) \Gamma( k+\frac{d-1}{2} ) }{  \Gamma(k+d-1)  } x^k  = \frac{ \Gamma( \frac{d-1}{2} )^2 }{ 2\Gamma(d-1) } {}_2F_1\Big( \frac{d-1}{2}, \frac{s}{2}; d-1; x\Big). 
\end{align*}
Here, we have used  
\begin{align*}
\int_0^{\pi/2} \sin^n(\theta) \cos^m(\theta)\,d\theta = \frac12 \frac{ \Gamma(\frac{n+1}{2}) \Gamma(\frac{m+1}{2} ) }{ \Gamma( \frac{m+n}{2}+1 ) }, \qquad (n,m>-1).  
\end{align*} 
Then we obtain  
\begin{align} \label{2.9}
 \int_{-1}^1   (1-t^2)^{ \frac{d-3}{2} } \Big( |y|^2+r^2-2|y|r t \Big)^{ -\frac{s}{2} }   \,dt  =  \frac{ 2^{d-2} }{ (|y|+r)^s }    \frac{ \Gamma( \frac{d-1}{2} )^2 }{  \Gamma(d-1) } {}_2F_1\Big( \frac{d-1}{2}, \frac{s}{2}; d-1;  \frac{4|y|r}{ (|y|+r)^2 }\Big).  
\end{align}
Combining \eqref{2.8} and \eqref{2.9}, the lemma follows. 
\end{proof}

\begin{rem} 
For the special case $d=2,s=1$, the hypergeometric functions can be expressed 
in terms of the complete elliptic integral
\begin{equation}
K(z) := \int_0^{\pi/2} \frac{1}{ \sqrt{ 1-z \sin^2\theta } } \,d\theta = \frac{\pi}{2} {}_2F_1( \tfrac12, \tfrac12; 1; z ). 
\end{equation} 
\end{rem}

Next, we evaluate the integral in \eqref{2.2}.

\begin{lem} \label{Lem_integral expansion gen}
For $|y|>1$, we have 
\begin{align}
\begin{split}\label{2.8a}
& \quad \int_{ 0 }^1    \frac{1}{(|y|+r)^s}    \, {}_2F_1\Big(  \frac{s}{2}, \frac{d-1}{2} \, ; \, d-1 \, ; \,   \frac{4r|y|}{(|y|+r)^2} \Big)\, f(r) \,dr 
\\
& =   \frac{ \Gamma( \frac{d}{2} ) }{ \Gamma(\frac{s}{2}) \Gamma( \frac{s-d+2}{2} )  }    \sum_{n=0}^\infty \frac{ \Gamma(\frac{s}{2}+n) \Gamma( \frac{s-d+2}{2}+n )  }{ \Gamma(\frac{d}{2}+n) \, n! } \int_0^1   r^{ 2n }  f(r)\,dr\, \frac{1}{|y|^{2n+s}}. 
\end{split}
\end{align}
For $|y|<1$, we have 
\begin{align}
\begin{split}
&\quad \int_{ 0 }^1    \frac{1}{(|y|+r)^s}    \, {}_2F_1\Big(  \frac{s}{2}, \frac{d-1}{2} \, ; \, d-1 \, ; \,   \frac{4r|y|}{(|y|+r)^2} \Big)\, f(r) \,dr 
\\
&=  \frac{ \Gamma( \frac{d}{2} ) }{ \Gamma(\frac{s}{2}) \Gamma( \frac{s-d+2}{2} )  }    \sum_{n=0}^\infty \frac{ \Gamma(\frac{s}{2}+n) \Gamma( \frac{s-d+2}{2}+n )  }{ \Gamma(\frac{d}{2}+n) \, n! } \int_0^{|y|}  r^{ 2n }  f(r)\,dr\, \frac{1}{|y|^{2n+s}}
\\
&\quad +  \frac{ \Gamma( \frac{d}{2} ) }{ \Gamma(\frac{s}{2}) \Gamma( \frac{s-d+2}{2} )  } \sum_{n=0}^\infty \frac{ \Gamma(\frac{s}{2}+n) \Gamma( \frac{s-d+2}{2}+n )  }{ \Gamma(\frac{d}{2}+n) \, n! }  \int_{|y|}^{1} \frac{1}{r^{2n+s}} f(r) \,dr \, |y|^{2n}. 
\end{split}
\end{align}
\end{lem}

\begin{proof}
For simplicity we write $w=|y|.$
By the change of variables $r \mapsto w t$, we have 
\begin{align*}
&\quad  \int_{ 0 }^1    \frac{1}{(w+r)^s}    \, {}_2F_1\Big(  \frac{s}{2}, \frac{d-1}{2} \, ; \, d-1 \, ; \,   \frac{4rw}{(w+r)^2} \Big)\, f(r) \,dr 
\\
& =  \frac{1}{w^{s-1}}  \int_0^{1/w}  \frac{1}{(1+ t)^s}    \, {}_2F_1\Big(  \frac{s}{2}, \frac{d-1}{2} \, ; \, d-1 \, ; \,   \frac{4t}{(1+ t)^2} \Big)\, f(wt)  \,dt.
\end{align*}
We use the quadratic transform: for $0\le z \le 1$,
\begin{equation} \label{quad trans 2F1}
{}_2F_1 \Big( a,b,2b, \frac{4z}{(1+z)^2} \Big) = (1+z)^{2a} {}_2F_1\Big( a,a+\frac12-b, b+\frac12, z^2 \Big); 
\end{equation}
see \cite[Eq.2.11(5)]{Bateman53}. 

By applying \eqref{quad trans 2F1}, for $w>1$, we have 
\begin{align*}
&\quad \frac{1}{w^{s-1}}  \int_0^{1/w}  \frac{1}{(1+ t)^s}    \, {}_2F_1\Big(  \frac{s}{2}, \frac{d-1}2 \, ; \, d-1 \, ; \,   \frac{4t}{(1+ t)^2} \Big)\, f(wt)  \,dt 
\\
&= \frac{1}{w^{s-1}}  \int_0^{1/w}  {}_2F_1\Big(  \frac{s}{2}, \frac{s-d+2}{2} \, ; \, \frac{d}{2} \, ; \,   t^2 \Big)\, f(wt)  \,dt
= \frac{1}{w^{s}}  \int_0^{1}  {}_2F_1\Big(  \frac{s}{2}, \frac{s-d+2}{2} \, ; \, \frac{d}{2}\, ; \,   \frac{r^2}{w^2} \Big)\, f(r)  \,dr 
\\
&=   \frac{ \Gamma( \frac{d}{2} ) }{ \Gamma(\frac{s}{2}) \Gamma( \frac{s-d+2}{2} )  }    \sum_{n=0}^\infty \frac{ \Gamma(\frac{s}{2}+n) \Gamma( \frac{s-d+2}{2}+n )  }{ \Gamma(\frac{d}{2}+n) \, n! } \int_0^1   r^{ 2n }  f(r)\,dr\, \frac{1}{w^{2n+s}}. 
\end{align*}
Here, we have used the definition \eqref{def of 2F1 Gauss series}.  

In contrast to the case $w>1$, for $w<1$ the application of the quadratic transformation \eqref{quad trans 2F1} requires splitting the integral:
\begin{align*}
&\quad \frac{1}{w^{s-1}}  \int_0^{1/w}  \frac{1}{(1+ t)^s}    \, {}_2F_1\Big(  \frac{s}{2}, \frac{d-1}{2} \, ; \, d-1 \, ; \,   \frac{4t}{(1+ t)^2} \Big)\, f(wt)  \,dt = \RN{1}+\RN{2},  
\end{align*}
where 
\begin{align*}
\RN{1} & =  \frac{1}{w^{s-1}}  \int_0^{1}  \frac{1}{(1+ t)^s}    \, {}_2F_1\Big(  \frac{s}{2}, \frac{d-1}{2} \, ; \, d-1 \, ; \,   \frac{4t}{(1+ t)^2} \Big)\, f(wt)  \,dt, 
\\
\RN{2} & =  \frac{1}{w^{s-1}}  \int_1^{1/w}   \frac{1}{(1+ t)^s}    \, {}_2F_1\Big(  \frac{s}{2}, \frac{d-1}{2} \, ; \, d-1 \, ; \,   \frac{4t}{(1+ t)^2} \Big)\, f(wt)  \,dt. 
\end{align*}
The first term can be handled as before, which leads to 
\begin{align*}
\RN{1} =   \frac{ \Gamma( \frac{d}{2} ) }{ \Gamma(\frac{s}{2}) \Gamma( \frac{s-d+2}{2} )  }    \sum_{n=0}^\infty \frac{ \Gamma(\frac{s}{2}+n) \Gamma( \frac{s-d+2}{2}+n )  }{ \Gamma(\frac{d}{2}+n) \, n! } \int_0^w  r^{ 2n }  f(r)\,dr\, \frac{1}{w^{2n+s}}.  
\end{align*}
For the second term, by first applying the change of variable $t \mapsto 1/u$ and then \eqref{quad trans 2F1}, we have 
\begin{align*}
\RN{2} &= \frac{1}{w^{s-1}}  \int_{w}^1  \frac{ u^{s-2} }{(1+ u)^s}    \, {}_2F_1\Big(  \frac{s}{2}, \frac{d-1}{2} \, ; \, d-1 \, ; \,   \frac{4u}{(1+ u)^2} \Big)\, f(w/u)  \,du
\\
&= \frac{1}{w^{s-1}} \int_w^1 u^{s-2}  {}_2F_1\Big(  \frac{s}{2}, \frac{s}{2} \, ; \, 1 \, ; \,   u^2 \Big)  f(w/u)\,du 
\\
&=  \frac{1}{w^{s-1}} \int_1^{1/w} \frac{1}{ t^{s} }  {}_2F_1\Big(  \frac{s}{2}, \frac{s-d+2}{2} \, ; \, \frac{d}{2} \, ; \,   \frac{1}{t^2} \Big)  f(wt)\,dt
=    \int_w^{1} \frac{ 1 }{ r^{s} }  {}_2F_1\Big(  \frac{s}{2}, \frac{s-d+2}{2} \, ; \, \frac{d}{2} \, ; \,   \frac{w^2}{r^2} \Big)  f(r)\,dr. 
\end{align*}
Thus it follows that 
\begin{align*}
\RN{2} = \frac{ \Gamma( \frac{d}{2} ) }{ \Gamma(\frac{s}{2}) \Gamma( \frac{s-d+2}{2} )  } \sum_{n=0}^\infty \frac{ \Gamma(\frac{s}{2}+n) \Gamma( \frac{s-d+2}{2}+n )  }{ \Gamma(\frac{d}{2}+n) \, n! }  \int_w^{1} \frac{1}{r^{2n+s}} f(r) \,dr \, w^{2n}. 
\end{align*}
Combining all of the above, we obtain the desired result. 
\end{proof}

\begin{rem}
In the limit $s \to d - 2$, $d > 2$, the RHS of (\ref{2.8a}) simplifies to $|y|^{-s}$, which is thus the evaluation of the LHS of (\ref{2.2}) for $|y| > 1$. This is consistent with Newton's theorem as noted above (\ref{e.4}).
\end{rem}

Next, we present an algebraic identity. 

\begin{lem} \label{Lem_algebraic equal}
Let $d \ge 1$ and $s \in (d-2,d)$. Then for any $k=0,1,\dots,$   
\begin{equation}
 \sum_{n=0}^\infty \frac{ \Gamma(\frac{s}{2}+n) \Gamma( \frac{s-d}{2}+n+1 )  }{ \Gamma(\frac{d}{2}+n) \, n! } \Big( \frac{1}{2n+2k+d}+\frac{1}{ 2n-2k+s-d } \Big) =0. 
\end{equation}
\end{lem}

\begin{proof} 
 
For $k=0,1,\dots,$ we define 
\begin{align}
F(k)&:= \sum_{n=0}^\infty \frac{ \Gamma(\frac{s}{2}+n) \Gamma( \frac{s-d}{2}+n+1 )  }{ \Gamma(\frac{d}{2}+n) \, n! }  \frac{1}{n+k+\frac{d}{2}}, 
\\
G(k)&:= - \sum_{n=0}^\infty \frac{ \Gamma(\frac{s}{2}+n) \Gamma( \frac{s-d}{2}+n+1 )  }{ \Gamma(\frac{d}{2}+n) \, n! }  \frac{1}{ n+\frac{s-d}{2}-k }  . 
\end{align} 
We claim that $F(k)=G(k)$. 

Note first that for $k=0,$
\begin{align*}
F(0) & =  \sum_{n=0}^\infty \frac{ \Gamma(\frac{s}{2}+n) \Gamma( \frac{s-d+2}{2}+n )  }{ \Gamma(\frac{d}{2}+n+1) \, n! }
 = \frac{ \Gamma(1+\frac{s-d}{2})\Gamma(\frac{s}{2}) }{ \Gamma(1+\frac{d}{2}) } {}_2F_1(\tfrac{s}{2},1+\tfrac{s-d}{2}; 1+\tfrac{d}{2};1)
=    \frac{ \Gamma(\tfrac{s}{2})\Gamma(1+\tfrac{s-d}{2}) \Gamma(d-s) }{ \Gamma(d-\frac{s}{2})\Gamma(1-\frac{s-d}{2}) },
\end{align*}
where we have used the fact that for $c-a-b>0$, 
\begin{equation} \label{def of Gauss sum}
{}_2F_1(a,b;c;1)= \frac{ \Gamma(c)\Gamma(c-a-b) }{ \Gamma(c-a)\Gamma(c-b) }; 
\end{equation}
see e.g. \cite[Eq.(15.4.20)]{NIST}. 
Similarly, we have 
\begin{align*}
G(0) & =  -\sum_{n=0}^\infty \frac{ \Gamma(\frac{s}{2}+n) \Gamma( \frac{s-d}{2}+n )  }{ \Gamma(\frac{d}{2}+n) \, n! } 
=- \frac{ \Gamma( \frac{s}{2} )\Gamma( \frac{s-d}{2} ) }{ \Gamma(\frac{d}{2}) } {}_2F_1(\tfrac{s}{2},\tfrac{s-d}{2};\tfrac{d}{2};1)
= -  \frac{ \Gamma( \tfrac{s}{2} )    \Gamma( \tfrac{s-d}{2} )  \Gamma(d-s) }{ \Gamma(\frac{d-s}{2})\Gamma( d-\frac{s}{2} ) }.
\end{align*}
Then by the reflection identity of the gamma function, we have $F(0)=G(0).$

Next, we consider the case $k \ge 1$. We proceed by induction and assume that $F(\ell)=G(\ell)$ holds for $\ell=0,1,\dots,k-1$.
Note that 
\begin{align}
\begin{split} 
\frac{1}{k!} \sum_{\ell=0}^k (-1)^\ell \binom{k}{\ell} \frac{1}{x+\ell} &=  \frac{1}{k!} \int_0^1 t^{x-1} \bigg( \sum_{\ell=0}^k (-1)^\ell \binom{k}{\ell} t^\ell\bigg) \,dt = \frac{1}{k!}\int_0^1 t^{x-1}(1-t)^{k} = \frac{\Gamma(x)}{\Gamma(x+k+1)}. 
\end{split}
\end{align} 
Using this, note that 
\begin{align*}
\frac{1}{k!} \sum_{\ell=0}^k (-1)^\ell \binom{k}{\ell} F(\ell) &=    \sum_{n=0}^\infty \frac{ \Gamma(\frac{s}{2}+n) \Gamma( \frac{s-d+2}{2}+n )  }{ \Gamma(\frac{d}{2}+n) \, n! } \bigg(  \frac{1}{k!} \sum_{\ell=0}^k (-1)^\ell \binom{k}{\ell}  \frac{1}{n+\ell+\frac{d}{2}} \bigg) 
\\
&= \sum_{n=0}^\infty \frac{ \Gamma(\frac{s}{2}+n) \Gamma( \frac{s-d+2}{2}+n )  }{ \Gamma(\frac{d}{2}+ n+k+1) \, n! } 
\\
&= \frac{ \Gamma(1+\frac{s-d}{2})\Gamma(\frac{s}{2}) }{ \Gamma(1+k+\frac{d}{2}) } {}_2F_1(\tfrac{s}{2},1+\tfrac{s-d}{2}; 1+\tfrac{d}{2}+k;1)
=  \frac{ \Gamma(\tfrac{s}{2})\Gamma(1+\tfrac{s-d}{2})  \Gamma(k+d-s)  }{    \Gamma(k+d-\frac{s}{2}) \Gamma(k+1-\frac{s-d}{2})  }. 
\end{align*}
Similarly, it follows that 
\begin{align*}
\frac{1}{k!} \sum_{\ell=0}^k (-1)^\ell \binom{k}{\ell} G(\ell) &= -     \sum_{n=0}^\infty \frac{ \Gamma(\frac{s}{2}+n) \Gamma( \frac{s-d}{2}+n+1 )  }{ \Gamma(\frac{d}{2}+n) \, n! } \bigg( \frac{1}{k!} \sum_{\ell=0}^k (-1)^\ell \binom{k}{\ell} \frac{1}{ n+\frac{s-d}{2}-\ell } \bigg)
\\
&= (-1)^k\sum_{n=0}^\infty \frac{ \Gamma(\frac{s}{2}+n) \Gamma( \frac{s-d}{2}+n-k )  }{ \Gamma(\frac{d}{2}+n) \, n! } 
\\
&=  (-1)^k\frac{ \Gamma( \frac{s}{2} )\Gamma( \frac{s-d}{2}-k ) }{ \Gamma(\frac{d}{2}) } {}_2F_1(\tfrac{s}{2},\tfrac{s-d}{2}-k;\tfrac{d}{2};1) =  (-1)^k  \frac{ \Gamma( \tfrac{s}{2} )    \Gamma( \tfrac{s-d}{2}-k ) \Gamma(k+d-s) }{ \Gamma(\frac{d-s}{2})\Gamma( k+d-\frac{s}{2} ) }. 
\end{align*}
Therefore, we can see that
\begin{equation}
\sum_{\ell=0}^k (-1)^\ell \binom{k}{\ell} F(\ell)= \sum_{\ell=0}^k (-1)^\ell \binom{k}{\ell} G(\ell).
\end{equation}
Since $F(\ell)=G(\ell)$ for all $\ell=0,1,\dots,k-1$ due the induction hypothesis, we conclude that $F(k)=G(k)$, which completes the proof. 
\end{proof}

\begin{rem}
By using the generalised Euler integral \cite[Eq.(16.5.2)]{NIST}
\begin{equation}
\int_0^1 u^{\alpha-1} (1-u)^{ \gamma-\alpha-1 } {}_2F_1(a_1,a_2;b_1;tu)\,du= \frac{ \Gamma(\alpha)\Gamma(\gamma-\alpha) }{ \Gamma(\gamma) } {}_3F_2( a_1,a_2,\alpha; b_1,\gamma;t), 
\end{equation}
we have 
\begin{align*}
F(k) &= \frac{ \Gamma(\frac{s}{2})\Gamma(1+\frac{s-d}{2}) }{ (k+\frac{d}{2}) \Gamma(\frac{d}{2}) }  {}_3F_2(  k+\tfrac{d}{2}, 1+\tfrac{s-d}{2}, \tfrac{s}{2};  \tfrac{d}{2}, \tfrac{d}{2}+k+1; 1  ),
\\
G(k)&= \frac{ \Gamma(\frac{s}{2})\Gamma(1+\frac{s-d}{2}) }{ (k+\frac{d-s}{2}) \Gamma(\frac{d}{2}) }  {}_3F_2(  \tfrac{s-d}{2}-k, 1+\tfrac{s-d}{2}, \tfrac{s}{2};  \tfrac{d}{2}, \tfrac{s-d}{2}+1-k; 1  ). 
\end{align*}
Therefore, the identity in Lemma~\ref{Lem_algebraic equal} can be expressed as \eqref{Id_3F2}.  
\end{rem}

Now, we are ready to prove Theorem~\ref{Thm_radial}. Due to the Euler-Lagrange conditions \eqref{def of EL eqn}, this follows from the following proposition.

\begin{prop}  \label{Prop_evaulation of the EL conditions} Let $d \ge 1$ and $s \in (d-2,d).$
Suppose that $\mu$ is of the form \eqref{def of eq msr mu ak gen d}.  
For $|y|>1$, we have 
\begin{align}
\begin{split}
 \int_{ |x|<1 } \frac{ d\mu(x) }{ |x-y|^s } = \frac{ \Gamma( \frac{d}{2} ) }{ \Gamma(\frac{s}{2}) \Gamma( \frac{s-d+2}{2} )  }    \sum_{n=0}^\infty \frac{ \Gamma(\frac{s}{2}+n) \Gamma( \frac{s-d+2}{2}+n )  }{ \Gamma(\frac{d}{2}+n) \, n! }  \bigg( \sum_{k=0}^\infty \frac{d\,a_k}{2k+d+2n} \bigg) \, \frac{1}{|y|^{2n+s}}. 
\end{split}
\end{align} 
For $|y|<1$, we have 
\begin{equation}\label{2.20}
 \int_{ |x| <1 } \frac{ d\mu(x)  }{  |x-y|^s }=   \frac{ \Gamma( \frac{d}{2} ) }{ \Gamma(\frac{s}{2}) \Gamma( \frac{s-d+2}{2} )  }   \sum_{n=0}^\infty \frac{ \Gamma(\frac{s}{2}+n) \Gamma( \frac{s-d+2}{2}+n )  }{ \Gamma(\frac{d}{2}+n) \, n! } \bigg( \sum_{k=0}^\infty \frac{d\,a_k}{ 2k+d-2n-s } \bigg) |y|^{2n} .
\end{equation} 
Furthermore, the Robin constant $c$ in \eqref{def of EL eqn} is given by
\begin{equation}
c=    \frac{ 2d }{ s }    \sum_{k=0}^\infty \frac{a_k}{ 2k+d-s }  . 
\end{equation}  
\end{prop}
\begin{proof}
By assumption, the radial density is given by \eqref{def of radial density ak}. Then the case $|y|>1$ follows immediately from Lemmas~\ref{Lem_angular integral gen} and \ref{Lem_integral expansion gen}. 

Next, we consider the case $|y| \in [0,1).$
Note that  by \eqref{def of radial density ak}, 
\begin{align*}
&\quad \sum_{n=0}^\infty \frac{ \Gamma(\frac{s}{2}+n) \Gamma( \frac{s-d+2}{2}+n )  }{ \Gamma(\frac{d}{2}+n) \, n! } \int_0^{|y|}  r^{ 2n }  f(r)\,dr\, \frac{1}{|y|^{2n+s}} 
\\
&=  \sum_{n=0}^\infty  \frac{ \Gamma(\frac{s}{2}+n) \Gamma( \frac{s-d+2}{2}+n )  }{ \Gamma(\frac{d}{2}+n) \, n! }  \sum_{k=0}^\infty  d\, a_k \int_0^{|y|}  r^{ 2n+2k+d-1 }  \,dr\, \frac{1}{|y|^{2n+s}} 
\\
&= \sum_{n=0}^\infty  \frac{ \Gamma(\frac{s}{2}+n) \Gamma( \frac{s-d+2}{2}+n )  }{ \Gamma(\frac{d}{2}+n) \, n! }  \sum_{k=0}^\infty \frac{ d\, a_k }{2n+2k+d} |y|^{ 2k+d-s } . 
\end{align*}
Also we have 
\begin{align*}
&\quad \sum_{n=0}^\infty \frac{ \Gamma(\frac{s}{2}+n) \Gamma( \frac{s-d+2}{2}+n )  }{ \Gamma(\frac{d}{2}+n) \, n! }  \int_{|y|}^{1} \frac{1}{r^{2n+s}} f(r) \,dr \, |y|^{2n}
\\
&= \sum_{n=0}^\infty \frac{ \Gamma(\frac{s}{2}+n) \Gamma( \frac{s-d+2}{2}+n )  }{ \Gamma(\frac{d}{2}+n) \, n! } \sum_{k=0}^\infty d\,a_k \int_{|y|}^1 r^{2k+d-1-2n-s} \,dr \, |y|^{2n}
\\
&= \sum_{n=0}^\infty \frac{ \Gamma(\frac{s}{2}+n) \Gamma( \frac{s-d+2}{2}+n )  }{ \Gamma(\frac{d}{2}+n) \, n! } \sum_{k=0}^\infty \frac{d\,a_k}{ 2k+d-2n-s } ( |y|^{2n}-|y|^{2k+d-s} ). 
\end{align*}
Combining these, we obtain
\begin{align*}
&\quad \sum_{n=0}^\infty \frac{ \Gamma(\frac{s}{2}+n) \Gamma( \frac{s-d+2}{2}+n )  }{ \Gamma(\frac{d}{2}+n) \, n! } \int_0^{|y|}  r^{ 2n }  f(r)\,dr\, \frac{1}{|y|^{2n+s}}  + \sum_{n=0}^\infty \frac{ \Gamma(\frac{s}{2}+n) \Gamma( \frac{s-d+2}{2}+n )  }{ \Gamma(\frac{d}{2}+n) \, n! }  \int_{|y|}^{1} \frac{1}{r^{2n+s}} f(r) \,dr \, |y|^{2n}
\\
&=  \sum_{n=0}^\infty \frac{ \Gamma(\frac{s}{2}+n) \Gamma( \frac{s-d+2}{2}+n )  }{ \Gamma(\frac{d}{2}+n) \, n! } \bigg( \sum_{k=0}^\infty \frac{d\,a_k}{ 2k+d-2n-s } \bigg) |y|^{2n} 
\\
&\quad + d\sum_{k=0}^\infty a_k \bigg( \sum_{n=0}^\infty \frac{ \Gamma(\frac{s}{2}+n) \Gamma( \frac{s-d+2}{2}+n )  }{ \Gamma(\frac{d}{2}+n) \, n! } \Big( \frac{1}{2n+2k+d}+\frac{1}{ 2n+s-2k-d } \Big)  \bigg) |y|^{2k+d-s}. 
\end{align*} 
Here, the second term vanishes due to Lemma~\ref{Lem_algebraic equal}. 
Then the resulting formula follows from Lemmas~\ref{Lem_angular integral gen} and \ref{Lem_integral expansion gen}.

We need a special treatment for the case $d=1$. Nevertheless, the required modification is minor, and the earlier reduction step is not required. For instance, when $|y|>1$ we have
\begin{align*}
\int_{ -1 }^1 \frac{d\mu(x)}{ |x-y|^s }  &= \frac12 |y|^{-s}\int_{ -1 }^1 \frac{  \sum_{k=0}^\infty a_k x^{2k}  }{ (1-x/y)^s }\,dx
= \frac12 |y|^{-s}\sum_{n=0}^\infty \frac{1}{y^n} \frac{ \Gamma(s+n) }{ \Gamma(s)\,n! } \int_{-1}^1 \sum_{k=0}^\infty a_k x^{2k+n} \,dx  
\\
&=  \frac12 |y|^{-s}\sum_{n=0}^\infty \frac{1}{y^{2n}} \frac{ \Gamma(s+2n) }{ \Gamma(s)\,(2n)! } \int_{-1}^1 \sum_{k=0}^\infty a_k x^{2(k+n)} \,dx  
= \sum_{n=0}^\infty \frac{ \Gamma(s+2n) }{ \Gamma(s)\,(2n)! } \bigg( \sum_{k=0}^\infty \frac{a_k}{2k+2n+1} \bigg) \frac{1}{|y|^{2n+s}} . 
\end{align*}
For $|y|<1$, we split the integral into two regions:
 \begin{align*}
\int_{ -1 }^1 \frac{d\mu(x)}{ |x-y|^s }  &= \frac12 |y|^{-s}\int_{ -1 }^1 \frac{  \sum_{k=0}^\infty a_k x^{2k}  }{ |1-x/y|^s }\,dx 
\\
&= \frac12 |y|^{-s}\int_{ (-|y|,|y|) } \frac{  \sum_{k=0}^\infty a_k x^{2k}  }{ |1-x/y|^s }\,dx   + \frac12 |y|^{-s}\int_{ (-1,1)\setminus (-|y|,|y|) } \frac{  \sum_{k=0}^\infty a_k x^{2k}  }{ |1-x/y|^s }\,dx .
\end{align*} 
The remaining computations follow in the same manner as above.

The evaluation of Robin's constant follows by setting $y=0$ in (\ref{2.20}).
\end{proof}

\begin{rem}
Making use of the identity of Lemma \ref{Lem_algebraic equal} we can check that both the integrals in Proposition \ref{Prop_evaulation of the EL conditions} agree if we set $|y|=1$. 
\end{rem}

\subsection{Proof of Theorem~\ref{Thm_eq msr power type}}

This subsection is devoted to the proof of Theorem~\ref{Thm_eq msr power type}.

We first show Lemma~\ref{Lem_ak for power type eq msr}. 

\begin{proof}[Proof of Lemma~\ref{Lem_ak for power type eq msr}]
Substituting (\ref{def of ak for power type eq msr}) in (\ref{def of V gen d}),
the first assertion follows from 
\begin{align}
\begin{split}
\label{ak sum v1}
\sum_{k=0}^\infty \frac{ a_k}{ 2n+s-2k-d } &= - \frac{ \Gamma(\alpha+1+\frac{d}{2}) }{  2 \Gamma(\frac{d}{2}+1) \Gamma(\alpha+1)}
\int_0^1 x^{ \frac{d-s}{2}-n -1} \Big ( \sum_{k=0}^\infty {\Gamma(-\alpha + k) \over\Gamma(-\alpha) k!} x^k \Big ) \, dx
\\
&= -  \frac{ \Gamma(\alpha+1+\frac{d}{2}) }{  \Gamma(\frac{d}{2}+1) }  \frac{   \Gamma( \frac{d-s}{2}-n ) }{2 \, \Gamma( \frac{d-s}{2}-n+\alpha+1  )}, 
\end{split}
\end{align}
where to obtain the second equality, the sum in the integral on the LHS has been evaluated according to binomial expansion formula (\ref{eq for binomial expansion}), the resulting integral has been evaluated as an example of the Euler beta integral, and finally the resulting sum is recognised from the hypergeometric function \eqref{def of 2F1 Gauss series}.
The polynomial potential \eqref{def of V power poly} follows from the polynomial expression of the hypergeometric function \cite[Eq.(15.2.4)]{NIST}. 
\end{proof}

We next show the following lemma, which will be used to establish the inequality (\ref{inequality for general}).

\begin{lem} \label{Lem_hypergeometric inequal}
Let $d \ge 1$ and $s \in [d-2,d)$. Then for any $m\in \mathbb{Z}_{ \ge 0 }$ and $x \ge 1$, we have 
\begin{equation}
 x^{-s}  {}_2F_1 ( \tfrac{s}{2},\tfrac{s-d+2}{2}; 2m+2+\tfrac{s}{2}; \tfrac{1}{x^2} ) -   \frac{ \Gamma(\frac{d-s}{2}) \Gamma(2m+2+\frac{s}{2}) }{ \Gamma(\frac{d}{2}) (2m+1)! }   {}_2F_1(\tfrac{s}{2},-2m-1;\tfrac{d}{2};x^2) \ge 0. 
\end{equation} 
\end{lem}
\begin{proof}  
Notice first that by \eqref{def of Gauss sum}, one can see that  equality holds for $x=1$. 
Recall the definition of the regularised hypergeometric function
\begin{equation}
{}_2\textup{\textbf{F}}_1(a,b;c;z):=\frac{1}{\Gamma(c)} {}_2F_1(a,b;c;z). 
\end{equation} 
Therefore, the desired inequality is equivalent to showing that for $x > 1$,  
\begin{align} \label{2.25}
 x^{-s}  {}_2 \textbf{\textup{F}} _1 ( \tfrac{s}{2},\tfrac{s-d}{2}+1; 2m+2+\tfrac{s}{2}; \tfrac{1}{x^2} ) -    \frac{ \Gamma(\frac{d-s}{2})  }{ (2m+1)! }   {}_2\textbf{\textup{F}}_1(\tfrac{s}{2},-2m-1;\tfrac{d}{2};x^2) \ge 0. 
\end{align} 

Note that for $z \in (0,1)$, by \cite[Eq.(15.8.4)]{NIST}, we have 
\begin{align*}
  \frac{ \sin( \pi(c-a-b) ) } { \pi }  {}_2\textbf{\textup{F}} _1(a,b;c;z)&=   \frac{1}{ \Gamma(c-a)\Gamma(c-b) }    {}_2\textbf{\textup{F}} _1(a,b;a+b-c+1;1-z)
   \\
   &\quad - \frac{ (1-z)^{c-a-b} }{ \Gamma(a)\Gamma(b) }     {}_2\textbf{\textup{F}} _1(c-a,c-b;c-a-b+1;1-z). 
\end{align*}
Using this, we obtain 
\begin{align}
\begin{split} \label{2.29}
&\quad x^{-s}  {}_2 \textbf{\textup{F}} _1 ( \tfrac{s}{2},\tfrac{s-d}{2}+1; 2m+2+\tfrac{s}{2}; \tfrac{1}{x^2} )  
\\
&=  x^{-s}  \frac{ \pi }{ \sin(\frac{s-d}{2}\pi) } \frac{1}{(2m+1)! \Gamma(2m+1+\frac{d}{2})} {}_2 \textbf{\textup{F}} _1 ( \tfrac{s}{2},\tfrac{s-d}{2}+1; -2m-\tfrac{d-s}{2}; 1-\tfrac{1}{x^2} )
\\
&\quad - x^{-s}  \frac{ \pi }{ \sin(\frac{s-d}{2}\pi) } \frac{ (1-\frac{1}{x^2})^{ 2m+1+\frac{d-s}{2} } }{ \Gamma(\tfrac{s}{2}) \Gamma(\tfrac{s-d}{2}+1 ) } {}_2 \textbf{\textup{F}} _1 ( 2m+2, 2m+1+\tfrac{d}{2}; 2m+2+\tfrac{d-s}{2}; 1-\tfrac{1}{x^2} ). 
\end{split}
\end{align}

Next, we claim that 
\begin{align} \label{2.28}
     {}_2\textbf{\textup{F}}_1(\tfrac{s}{2},-2m-1;\tfrac{d}{2};x^2)
= x^{-s}  \frac{ \pi }{ \sin(\frac{s-d}{2}\pi) } \frac{1}{   \Gamma(\tfrac{d-s}{2}) \Gamma(2m+1+\frac{d}{2})} {}_2 \textbf{\textup{F}} _1 ( \tfrac{s}{2},\tfrac{s-d}{2}+1; -2m-\tfrac{d-s}{2}; 1-\tfrac{1}{x^2} ). 
\end{align}
By using \cite[Eq.(15.8.7)]{NIST}, we have 
\begin{align*}
  {}_2\textbf{\textup{F}}_1(\tfrac{s}{2},-2m-1;\tfrac{d}{2};x^2) = \frac{\pi}{ \sin( \frac{s-d}{2}\pi ) }  \frac{ x^{4m+2} }{\Gamma(\frac{d-s}{2})\Gamma(2m+1+\frac{d}{2})}   {}_2 \textbf{\textup{F}} _1(  -2m-1, -2m-\tfrac{d}{2}; -2m+\tfrac{s-d}{2};1-\tfrac{1}{x^2} ). 
\end{align*}
Therefore, the claim reduces to 
\begin{align*}
x^{-s}  {}_2 \textbf{\textup{F}} _1 ( \tfrac{s}{2},\tfrac{s-d}{2}+1; -2m-\tfrac{d-s}{2}; 1-\tfrac{1}{x^2} )= x^{4m+2}  {}_2 \textbf{\textup{F}} _1(  -2m-1, -2m-\tfrac{d}{2}; -2m+\tfrac{s-d}{2};1-\tfrac{1}{x^2} ). 
\end{align*}
This follows from Euler's transformation \cite[Eq.(15.8.1)]{NIST}. 

By combining \eqref{2.29} and \eqref{2.28}, we have shown that 
\begin{align}
\begin{split}
&\quad x^{-s}  {}_2 \textbf{\textup{F}} _1 ( \tfrac{s}{2},\tfrac{s-d}{2}+1; 2m+2+\tfrac{s}{2}; \tfrac{1}{x^2} ) -    \frac{ \Gamma(\frac{d-s}{2})  }{ (2m+1)! }   {}_2\textbf{\textup{F}}_1(\tfrac{s}{2},-2m-1;\tfrac{d}{2};x^2)
\\
&= x^{-s}  \frac{ \pi }{ \sin(\frac{d-s}{2}\pi) } \frac{ (1-\frac{1}{x^2})^{ 2m+1+\frac{d-s}{2} } }{ \Gamma(\tfrac{s}{2}) \Gamma(\tfrac{s-d}{2}+1 ) } {}_2 \textbf{\textup{F}} _1 ( 2m+2, 2m+1+\tfrac{d}{2}; 2m+2+\tfrac{d-s}{2}; 1-\tfrac{1}{x^2} ). 
\end{split}
\end{align}
Since $s \in (d-2,d)$, we have $\sin(\frac{d-s}{2}\pi)>0$. Moreover, as $x>1$, to prove \eqref{2.25} it is equivalent to show that for $y \in (0,1)$,
\begin{equation} \label{2.27}
  {}_2 \textbf{\textup{F}} _1 ( 2m+1+\tfrac{d}{2},2m+2; 2m+2+\tfrac{d-s}{2}; y )\ge 0. 
\end{equation}
We use the integral representation: for $c>b>0,$
\begin{equation} \label{int rep for 2F1}
 {}_2 \textbf{\textup{F}} _1 (a,b;c;z)= \frac{1}{\Gamma(b)\Gamma(c-b)} \int_0^1 \frac{ t^{b-1}(1-t)^{c-b-1} }{ (1-z t)^a }\,dt. 
\end{equation}
From this representation, one immediately sees that \eqref{2.27} holds, since the resulting integrand is positive. 
\end{proof}

We are now ready to prove Theorem~\ref{Thm_eq msr power type}. 

\begin{proof}[Proof of Theorem~\ref{Thm_eq msr power type}]
Let $a_k$ be given by \eqref{def of ak for power type eq msr} with $\alpha$ given by \eqref{def of alpha for power}.  
Then by \eqref{def of eq msr mu ak gen d} and \eqref{eq for binomial expansion}, the associated measure $\mu$ is given by \eqref{def of power type eq msr}.  
Combining \eqref{ak sum v1}, together with the definition \eqref{def of 2F1 Gauss series} of the hypergeometric function, it follows that the inequality \eqref{inequality for general} is equivalent to that in Lemma~\ref{Lem_hypergeometric inequal}. This completes the proof of the equilibrium measure.  

In order to show \eqref{energy for power type eq msr}, note that by taking $n=0$ in \eqref{ak sum v1}, we have  
\begin{align*}
\frac{d}{s}\sum_{k=0}^\infty \frac{ a_k}{ 2k+d-s }  =   \frac{1}{s} \frac{ \Gamma( \frac{s}{2}+2m+2) }{  \Gamma(\frac{d}{2}) }  \frac{   \Gamma( \frac{d-s}{2} ) }{  (2m+1)!} . 
\end{align*}
Also we have 
\begin{align*}
&\quad \frac12 \int_{ \R^d } V(x) \,d\mu(x) = \frac{ \Gamma(\alpha+1+\frac{d}{2}) }{ \Gamma(\alpha+1)\Gamma(\frac{d}{2}) } \sum_{k=1}^{2m+1} b_k \int_0^1 (1-r^2)^\alpha  r^{2k+d-1}\,dr  = \frac12\frac{ \Gamma(\alpha+1+\frac{d}{2}) }{  \Gamma(\frac{d}{2}) } \sum_{k=1}^{2m+1} b_k \frac{ \Gamma(\frac{d}{2}+k) } { \Gamma(\alpha+1+\frac{d}{2}+k) }  
\\
&=       \frac{ \Gamma(\frac{d-s}{2})  \Gamma(2m+2+\frac{s}{2} ) }{  \Gamma(\frac{d}{2})   }   \sum_{k=1}^{2m+1}  \frac{ (-1)^{k+1}  }{  k!(2m+1-k)!  } \frac{(s+2)(s+4)\dots(s+2k-2)}{ (s+4m+4)(s+4m+6)\dots(s+4m+2k+2) }. 
\end{align*}
Then by \eqref{energy in terms of Robin}, we obtain \eqref{energy for power type eq msr}.  
\end{proof}

\subsection{Proof of Theorem~\ref{Thm_eq msr for power potential}}
  
We first show Lemma~\ref{Lem_ak for power type potential}. 

\begin{proof}[Proof of Lemma~\ref{Lem_ak for power type potential}]
We claim that for $t \in (0,1),$
\begin{align}\label{2.32}
 \sum_{k=0}^\infty \frac{ \Gamma(k+t) }{ (k+t-p)k! } \frac{1}{ n-k-t }= -\frac{\pi}{ \sin(\pi t) }  \frac{1}{ n-p  } \Big(  \frac{ \Gamma( t-n ) }{ \Gamma(1-n) }-\frac{ \Gamma( t-p ) }{ \Gamma(1-p) } \Big).
\end{align}
Note that 
\begin{align*}
 \sum_{k=0}^\infty \frac{ \Gamma(k+t) }{ (k+t-p)k! } \frac{1}{ n-k-t } &= \frac{1}{n-p} \sum_{k=0}^\infty \frac{ \Gamma(k+t) }{ k! } \Big(\frac{1}{k+t-p}-\frac{1}{ k+t-n }\Big) 
\\
&= \frac{1}{n-p} \sum_{k=0}^\infty \frac{ \Gamma(k+t) }{ k! } \Big(\frac{ \Gamma(k+t-p) }{\Gamma(k+t-p+1)}-\frac{\Gamma(k+t-n)}{ \Gamma(k+t-n+1) }\Big) .
\end{align*} 
Then the desired identity \eqref{2.32} follows from \eqref{def of 2F1 Gauss series}, \eqref{def of Gauss sum} and the reflection property of the gamma function. 

By putting $t=\frac{d-s}{2}$ in \eqref{2.32}, we have 
\begin{align} \label{2.33}
 \sum_{k=0}^\infty \frac{ a_k}{ 2n+s-2k-d } &=    \frac{\Gamma(1+\frac{s}{2})}{ \Gamma( \frac{d}{2} ) } \, \frac{2p+s}{ 2d( n-p ) } \Big(  \frac{ \Gamma( \frac{d-s}{2}-n ) }{ \Gamma(1-n) }-\frac{ \Gamma( \frac{d-s}{2}-p ) }{ \Gamma(1-p) } \Big).
\end{align}
Note that as $n \to p$,  
\begin{align*}
\frac{1}{n-p}  \Big(  \frac{ \Gamma( \frac{d-s}{2}-n ) }{ \Gamma(1-n) }-\frac{ \Gamma( \frac{d-s}{2}-p ) }{ \Gamma(1-p) } \Big)  \to \Gamma(\tfrac{d-s}{2}-p) (p-1)! (-1)^p. 
\end{align*}
By using this, the expression \eqref{def of V for power type} follows from \eqref{def of V gen d}. 
\end{proof}

Next we show the following inequality.

\begin{lem}\label{Lem_hypergeometric inequal for power type potential}
Let $d \ge 1$ and $s \in [d-2,d)$. Then for any $p\in \mathbb{Z}_{ \ge 0 }$ and $x \ge 1$, we have 
\begin{align} \label{2.34}
 \frac{\sin(\pi\tfrac{d-s}{2})}{\pi} \sum_{n=0}^\infty \frac{  \Gamma( \frac{s-d}{2}+n+1 )  }{  (n+\frac{s}{2}) n! }   \frac{2p}{2n+2p+s}    \, \frac{1}{x^{2n+s}}  \ge     \frac{ \Gamma(\frac{s}{2})   }{ \Gamma(\frac{d}{2})  }  \,  -    \frac{ \Gamma(\frac{s}{2}+p)   }{ \Gamma(\frac{d}{2}+p)  } \,      x^{2p}.
\end{align}
\end{lem}
\begin{proof} 

We first consider the case $s >0$. Then by \eqref{eq for binomial expansion}, we have 
\begin{align*}
 \sum_{n=0}^\infty \frac{  \Gamma( \frac{s-d}{2}+n+1 )  }{  (n+\frac{s}{2}) n! }   \frac{2p}{2n+2p+s} y^{n} &= 2\sum_{n=0}^\infty \frac{  \Gamma( \frac{s-d}{2}+n+1 )  }{    n! }  \bigg(\int_0^1 u^{2n+s-1}(1-u^{2p})\,du \bigg)  y^n
\\
&= 2 \Gamma( 1+\tfrac{s-d}{2} ) \int_0^1 u^{s-1} (1-u^2y)^{ \frac{d-s}{2}-1 } (1-u^{2p})\,du .
\end{align*}
Here, the condition $s>0$ is required for the integrability when $n=0$.
Then, the LHS of \eqref{2.34} is written as 
\begin{align}
\begin{split}
 \frac{2}{ \Gamma( \frac{d-s}{2} ) } \frac{1}{x^s} \int_0^1 u^{s-1} (1-u^2/x^2)^{ \frac{d-s}{2}-1 } (1-u^{2p})\,du
= \frac{2}{ \Gamma( \frac{d-s}{2} ) }   \int_0^{1/x} t^{s-1} (1-t^2)^{ \frac{d-s}{2}-1 } (1-t^{2p}x^{2p})\,dt. 
\end{split}  
\end{align}
Note that for $t>1/x$, we have $1-t^{2p}x^{2p}<0$. Therefore 
\begin{align*}
&\quad \frac{2}{ \Gamma( \frac{d-s}{2} ) }   \int_0^{1/x} t^{s-1} (1-t^2)^{ \frac{d-s}{2}-1 } (1-t^{2p}x^{2p})\,dt 
\\
& \ge \frac{2}{ \Gamma( \frac{d-s}{2} ) }   \int_0^{1} t^{s-1} (1-t^2)^{ \frac{d-s}{2}-1 } (1-t^{2p}x^{2p})\,dt =   \frac{ \Gamma(\frac{s}{2})   }{ \Gamma(\frac{d}{2})  }  \,  -    \frac{ \Gamma(\frac{s}{2}+p)   }{ \Gamma(\frac{d}{2}+p)  } \,      x^{2p},
\end{align*}
where the equality follows from the evaluation of the Euler beta integral in terms of gamma functions.
This completes the proof for the case $s>0$. 

For the remaining case $s \in (-1,0]$ when $d=1$, a slight modification is required to ensure integrability. First, note that 
\begin{align*}
&\quad \sum_{n=0}^\infty \frac{  \Gamma( \frac{s-d}{2}+n+1 )  }{  (n+\frac{s}{2}) n! }   \frac{2p}{2n+2p+s} y^{n} = \frac{  \Gamma( \frac{s-d}{2}+1 )  }{  s }   \frac{4p}{2p+s}+  \sum_{n=1}^\infty \frac{  \Gamma( \frac{s-d}{2}+n+1 )  }{  (n+\frac{s}{2}) n! }   \frac{2p}{2n+2p+s} y^{n} 
\\
&= \frac{  \Gamma( \frac{s-d}{2}+1 )  }{  s }   \frac{4p}{2p+s}+  2\sum_{n=1}^\infty \frac{  \Gamma( \frac{s-d}{2}+n+1 )  }{    n! }  \bigg(\int_0^1 u^{2n+s-1}(1-u^{2p})\,du \bigg)  y^n
\\
&=\frac{  \Gamma( \frac{s-d}{2}+1 )  }{  s }   \frac{4p}{2p+s}+ 2 \Gamma( \tfrac{s-d}{2}+1 ) \int_0^1 u^{s-1} \Big( (1-u^2y)^{ \frac{d-s}{2}-1 }-1\Big) (1-u^{2p})\,du .
\end{align*} 
Then the LHS of \eqref{2.34} is written as 
\begin{align}
\begin{split}
& \quad  \frac{2}{ \Gamma( \frac{d-s}{2} ) }  \frac{2p}{s(2p+s)} \frac{1}{x^s} + \frac{2}{ \Gamma( \frac{d-s}{2} ) } \frac{1}{x^s} \int_0^1 u^{s-1} \Big( (1-u^2/x^2)^{ \frac{d-s}{2}-1 } -1 \Big) (1-u^{2p})\,du
\\
&=  \frac{2}{ \Gamma( \frac{d-s}{2} ) }  \frac{2p}{s(2p+s)} \frac{1}{x^s} +  \frac{2}{ \Gamma( \frac{d-s}{2} ) }   \int_0^{1/x} t^{s-1} \Big( (1-t^2)^{ \frac{d-s}{2}-1 }-1\Big) (1-t^{2p}x^{2p})\,dt. 
\end{split}  
\end{align}
Here, again by using the Euler beta integral, we obtain  
\begin{align*}
&\quad \frac{2}{ \Gamma( \frac{d-s}{2} ) }   \int_0^{1/x} t^{s-1} \Big( (1-t^2)^{ \frac{d-s}{2}-1 }-1\Big) (1-t^{2p}x^{2p})\,dt 
\\
&\ge \frac{2}{ \Gamma( \frac{d-s}{2} ) }   \int_0^{1} t^{s-1} \Big( (1-t^2)^{ \frac{d-s}{2}-1 }-1\Big) (1-t^{2p}x^{2p})\,dt 
= \frac{2}{ \Gamma( \frac{d-s}{2} ) }\Big( \frac{x^{2p}}{s+2p} -\frac{1}{s}\Big)   + \frac{ \Gamma(\frac{s}{2})   }{ \Gamma(\frac{d}{2})  }   -    \frac{ \Gamma(\frac{s}{2}+p)   }{ \Gamma(\frac{d}{2}+p)  } \,      x^{2p}.
\end{align*}
Therefore, we have shown that 
\begin{align*}
&\quad  \frac{\sin(\pi\tfrac{d-s}{2})}{\pi} \sum_{n=0}^\infty \frac{  \Gamma( \frac{s-d}{2}+n+1 )  }{  (n+\frac{s}{2}) n! }   \frac{2p}{2n+2p+s}    \, \frac{1}{x^{2n+s}}
\\
& \ge \frac{2}{ \Gamma( \frac{d-s}{2} ) } \frac{1}{s(2p+s)} \Big( 2p(x^{-s}-1)+s (x^{2p}-1) \Big)  + \frac{ \Gamma(\frac{s}{2})   }{ \Gamma(\frac{d}{2})  }   -    \frac{ \Gamma(\frac{s}{2}+p)   }{ \Gamma(\frac{d}{2}+p)  } \,      x^{2p}.
\end{align*}
Since $s<0$ and $p>0$, the first term in the RHS is nonnegative for $x \ge 1$. Therefore the proof is complete. 
\end{proof}

\begin{proof}[Proof of Theorem~\ref{Thm_eq msr for power potential}]

All we need to check is the inequality \eqref{inequality for general}. 
Note that by \eqref{2.32}, we have 
\begin{align*}
\sum_{k=0}^\infty   \frac{ \Gamma(k+\frac{d-s}{2}) }{ (k+\frac{d-s}{2}-p)k! }  \frac{1}{2k+d+2n} = -\frac{\pi}{ \sin(\pi \frac{d-s}{2})  } \frac{1}{2n+2p+s} \frac{ \Gamma(n+\frac{d}{2}) }{ \Gamma(n+1+\frac{s}{2}) }. 
\end{align*}
Thus 
\begin{align*}
\sum_{k=0}^\infty \frac{a_k}{2k+d+2n}  =     \frac{\Gamma(1+\frac{s}{2})}{ \Gamma( \frac{d}{2} ) } (2p+s)\,   \frac{1}{2n+2p+s} \frac{ \Gamma(n+\frac{d}{2}) }{ \Gamma(n+1+\frac{s}{2}) } 
\end{align*}
and 
\begin{align}
\begin{split}
& \quad  \sum_{n=0}^\infty \frac{ \Gamma(\frac{s}{2}+n) \Gamma( \frac{s-d}{2}+n+1 )  }{ \Gamma(\frac{d}{2}+n) \, n! }  \bigg( \sum_{k=0}^\infty \frac{a_k}{2k+d+2n} \bigg) \, \frac{1}{x^{2n+s}}
\\
&=   \frac{\Gamma(1+\frac{s}{2})}{ \Gamma( \frac{d}{2} ) } (2p+s)\,   \sum_{n=0}^\infty \frac{  \Gamma( \frac{s-d}{2}+n+1 )  }{  (n+\frac{s}{2}) n! }   \frac{1}{2n+2p+s}    \, \frac{1}{x^{2n+s}}.
\end{split}
\end{align}

On the other hand, by \eqref{2.33} and Lemma~\ref{Lem_ak for power type potential}, we have  
\begin{align}
\begin{split}
&\quad \sum_{n=0}^\infty \frac{ \Gamma(\frac{s}{2}+n) \Gamma( \frac{s-d}{2}+n+1 )  }{ \Gamma(\frac{d}{2}+n) \, n! } \bigg( \sum_{k=0}^\infty \frac{a_k}{ 2k+d-2n-s } \bigg) x^{2n} 
\\
&=  \frac{ \Gamma(\frac{s}{2}) \Gamma( \frac{s-d}{2}+1 )  }{ \Gamma(\frac{d}{2})  } \bigg( \sum_{k=0}^\infty \frac{a_k}{ 2k+d-s } \bigg)  +\frac{ \Gamma(\frac{s}{2}+1) \Gamma( \frac{s-d}{2}+1 )  }{ 2\, \Gamma( \frac{d}{2}+1 ) }  V(x)
\\
&=  \frac{\pi}{\sin(\pi\tfrac{d-s}{2})}  \frac{ \Gamma(\frac{s}{2})   }{ \Gamma(\frac{d}{2})  }  \frac{\Gamma(1+\frac{s}{2})}{ \Gamma( \frac{d}{2} ) } \, \frac{2p+s}{ 2p }   - \frac{\pi}{\sin(\pi\tfrac{d-s}{2})} \frac{ \Gamma(\frac{s}{2}+1)    }{  \Gamma( \frac{d}{2}+1 ) }     \frac{ \Gamma(\frac{s}{2}+p)   }{ \Gamma(\frac{d}{2}+p)  } \, \frac{d(2p+s)}{ 4p  }       |x|^{2p}.
\end{split}
\end{align}
Combining the above, one can see that the inequality \eqref{inequality for general} is equivalent to that in Lemma~\ref{Lem_hypergeometric inequal for power type potential}.

Note that by \eqref{2.33}, we have 
\begin{equation}
\frac{d}{s} \sum_{k=0}^\infty \frac{ a_k}{ 2k+d-s }  =  \frac{ \Gamma( \tfrac{d-s}{2} ) \Gamma(\frac{s}{2})}{ \Gamma( \frac{d}{2} ) } \, \frac{2p+s}{ 4 p  }.
\end{equation} 
Note that by \eqref{def of 2F1 Gauss series}, 
\begin{align*}
&\quad \int_0^1  {}_2F_1( \tfrac{d-s}{2}, \tfrac{d-s}{2}-p; \tfrac{d-s}{2}+1-p; r^2 )r^{2p+d-1}\,dr 
\\
&= \frac{ (d-s-2p) }{ \Gamma(\tfrac{d-s}{2})   } \sum_{k=0}^\infty \int_0^1 \frac{ \Gamma(\tfrac{d-s}{2}+k)  }{ (d-s-2p+2k) \,k! } r^{2k+2p+d-1}\,dr  =    \frac{  2p-d+s }{ 2(4p+s) }   \frac{  \Gamma(1+\tfrac{s-d}{2}) \Gamma(p+\frac{d}{2}) }{ \Gamma(p+\frac{s}{2}+1) }  . 
\end{align*}
Then by \eqref{def of V for power type} and \eqref{def of power type potential eq msr}, we have 
\begin{align*}
\frac12 \int_{ \R^d } V(x) \,d\mu(x)  =   \frac{s(2p+s)}{ 4 p (4p+s) }  \frac{\Gamma(\frac{d-s}{2})  \Gamma(\tfrac{s}{2})  } {  \Gamma(\frac{d}{2}) }    . 
\end{align*} 
Therefore by \eqref{energy in terms of Robin}, we obtain \eqref{energy for power type potential}.   
\end{proof}

\section{Proof of Theorem~\ref{Thm_Coulomb half space}}\label{S3}
 
In this section, we prove Theorem~\ref{Thm_Coulomb half space}.
In Subsection~\ref{S41}, we revisit the one-dimensional case \( d=1 \), where all computations can be made explicit.
Subsection~\ref{S42} is devoted to the general high-dimensional case, where we complete the proof of Theorem~\ref{Thm_Coulomb half space}, with the exception of \eqref{LD Prob in main}.
Finally, the statement concerning the large deviation probability \eqref{LD Prob in main} is established in Subsection~\ref{S43}.

\subsection{The two-dimensional case revisited} \label{S41}

It is instructive to recall the derivation for the case \( d = 1 \) from \cite{ASZ14}, where Conjecture \ref{Cj1} is a theorem, and one can make use of several well-known computations from logarithmic potential theory. 

Note that the variational condition is given by
\begin{equation}\label{3.11}
2 \int_{\R^2} \log\frac{1}{|x-y|} d\mu_V(y) + |x|^2  \begin{cases}
= c & \textup{if } x_0 \ge a = \sqrt{2},\,|x_1|<\sqrt{2},
\smallskip
\\
\ge c & \textup{otherwise}. 
\end{cases} 
\end{equation} 
For $t, x \in \R$, let 
\begin{align}\label{3.12}
\begin{split}
G(t,x)&:= 2 \int_{\R^2} \log\frac{1}{|(a+t,x)-\widehat{y}|} d\mu_V(\hat{y}) +(a+t)^2+ x^2 
\\
&=  2 \int_{ |y|<\sqrt{2} } \log \frac{1}{ |(a+t,x)-(a,y)| } \frac{ \sqrt{2-y^2} }{\pi}\,dy  +(a+t)^2+ x^2.
\end{split}
\end{align}
Hence $G(t,x)$ denotes the potential as specified by the LHS of (\ref{3.11}), with $t$ the horizontal coordinate with value zero at the position of the plane $x_0=a$, and $x$ the vertical coordinate
corresponding to $x_1$.

It is well known that \cite{ST97}
\begin{align} \label{EL condition for semi circle}
2 \int_\R \log \frac{1}{|u-v|} \frac{ \sqrt{2-v^2} }{ \pi } \mathbbm{1}_{ \{ |v| < \sqrt{2} \}  } \,dv   +u^2-1- \log 2  \begin{cases}
= 0, & |u| \le \sqrt{2},
\smallskip
\\
\ge 0 & |u| > \sqrt{2}, 
\end{cases}\qquad u \in \R.
\end{align}
Therefore if $t=0$ in (\ref{3.12}) we have
\begin{align*}
G(0,x) \begin{cases}
= 1+\log 2 +a^2 &\textup{if } |x| \le \sqrt{2},
\smallskip 
\\
\ge 1+\log 2 +a^2 &\textup{if } |x| \ge \sqrt{2}. 
\end{cases}
\end{align*} 
To complete the verification of the variational condition,
it now suffices to show that for any $t \ge 0$ and $x \in \R$, 
\begin{equation}\label{3.14}
G(t,x) \ge G(0,0)= 1+\log 2 +a^2. 
\end{equation}

In preparation for verifying (\ref{3.14}), using complex variables we write 
\begin{align} \label{3.15}
 2 \int_{ |y|<\sqrt{2} } \log \frac{1}{ |(a+t,x)-(a,y)| } \frac{ \sqrt{2-y^2} }{\pi}\,dy =  2 \int_{ |y|<\sqrt{2} } \log \frac{1}{ |t+i(x-y)| } \frac{ \sqrt{2-y^2} }{\pi}\,dy. 
\end{align}
Then 
\begin{align*}
\frac{d}{dt} G(t,x)&= 2(a+t)  -2 \int_\R \re \Big[ \frac{1}{t+i(x-y)} \Big]  \frac{ \sqrt{ 2-y^2 } }{ \pi } \cdot \mathbbm{1}_{ \{ |y|<\sqrt{ 2 } \} } \,dy
\\
&=  2(a+t)  -2  \int_\R \im \Big[ \frac{1}{u-z} \Big]   \frac{ \sqrt{ 2-u^2 } }{ \pi } \cdot \mathbbm{1}_{ \{ |u|<\sqrt{ 2 } \} } \,du  ,
\end{align*}
where  $z=-x+it$.  
Another well-known integral identity related to the semicircle law is its Stieltjes transform 
\begin{align}\label{3.16+}
\int_{\mathbb R} \frac{1}{u-z} \frac{ \sqrt{2-u^2} }{ \pi } \mathbbm{1}_{ \{ |u| < \sqrt{2} \}  } \,du= -z+\sqrt{z^2-2}, \qquad z \in \C_+. 
\end{align}
Using this in the previous equation shows
\begin{align}\label{3.16a}
\frac{d}{dt} G(t,x) &= 2(a+t)  - 2 \im \Big[ x-it+\sqrt{ (x-it)^2-2 } \Big] = 2(a+t) +2  \Big( t- \im  \sqrt{ (x-it)^2-2 } \Big). 
\end{align} 
Observe that this is precisely where the threshold \( a > \sqrt{2} \) arises.  
Namely, for the measure to be the correct equilibrium measure, one requires  
$$
\frac{d}{dt} G(t,x)\Big|_{t=0}= 2a-2\im \sqrt{x^2-2} 
$$ 
to be positive. In particular, by taking $x^2$ as small as possible and thus \( x = 0 \), this condition yields the requirement \( a > \sqrt{2}= a_{\rm cri} \), which moreover establishes statement (i) of
Theorem \ref{Thm_Coulomb half space}.

We will make use of (\ref{3.16a}) to verify (\ref{3.14}), which
we do in two steps.
\begin{itemize}
    \item (Step 1). We first show that for any $t \ge 0$, $G(t,0) \ge G(0,0)$. This follows from the fact that for $a>\sqrt{2} = a_{\rm cri}$,
    \begin{align*}
\frac{d}{dt} G(t,0)  = 2(a+t) +2  \Big( t-  \sqrt{ t^2+2 } \Big) \ge 0.  
\end{align*}
   \item (Step 2). Next, we show that for any fixed $t>0$, $G(t,x) \ge G(t,0)$. This follows from a similar computation above, which gives
   \begin{align*}
 \frac{d}{dx} G(t,x) &= 2x - 2 \int_{ \R } \re\Big[ \frac{i}{ t+i(x-y) }\Big] \frac{\sqrt{2-y^2}}{ \pi }\cdot \mathbbm{1}_{ \{ |y|<\sqrt{ 2 } \} } \,dy
 \\
 &= 2x + 2 \int_{ \R } \re \Big[ \frac{1}{y-(x+it)} \Big] \frac{ \sqrt{2-u^2} }{\pi}\cdot \mathbbm{1}_{ \{ |y|<\sqrt{ 2 } \} } \,du  
 \\
 &= 2x + 2 \re  \Big( -(x+it)+\sqrt{ (x+it)^2-2 } \Big)  =  2 \re  \sqrt{ (x-it)^2-2 } \ge 0 .  
   \end{align*}
   Here, we have used 
   \begin{align*}
    \re\Big[ \frac{i}{ t+i(x-y) }\Big] =  \re \Big[ \frac{1}{y-(x+it)} \Big] = \frac{y-x}{t^2+(x-y)^2},  
   \end{align*}
   as well as (\ref{3.16+}) in going from the first to the third
   equality. The final inequality has made use of a general property of the square root function with branch cut on the negative real axis.
\end{itemize}
Combining these steps, it follows that for any $t \ge 0$ and $x \in \R$, $G(t,x)\ge G(0,0)$, which completes the proof. 

\begin{rem} \label{Rem_elliptic half space}
Let $x=(x_0,x_1)$. With $\mu \in \mathbb R$ a parameters, $c_0^2:= 2/(1 + e^{2 \mu})$, $c_1^2:= 2/(1 + e^{-2 \mu})$, a generalisation of the potential (\ref{V1}) is
\begin{equation}\label{V1a}
V({x}) = \begin{cases}
(c_0 x_0)^2 + (c_1 x_1)^2, &\textup{if } x_0 > a,
\smallskip 
\\
+\infty, &\textup{otherwise}. 
\end{cases}
\end{equation}
In the case $a = - \infty$ (no hard wall) the equilibrium measure is the ellipse shaped droplet \cite{CPR87,FJ96}
\begin{equation}\label{V2a}
\sigma(x) = {1 \over \pi} \mathbbm 1_{\{ (e^{-\mu} x_0)^2 +
(e^{\mu} x_1)^2 <1 \}}.
\end{equation}
A straightforward extension of the above analysis, making use of scaled versions of (\ref{EL condition for semi circle}) and (\ref{3.16+}), shows that the threshold for the equilibrium measure to be fully supported on $x_0=a$ is $a \ge  \sqrt{2} c_1/c_0^2$, and it then has density on this line
\begin{equation}\label{V2b}
\mu(x_1) = {\sqrt{2 - (x_1/c_1)^2} \over \pi c_1} \mathbbm 1_{\{|x_1/c_1| < \sqrt{2} \}}.
\end{equation} 
Note that this conclusion holds independently of the parameterisation of $c_0,c_1$. In particular, this allows us to consider $c_1 \to \infty$, when the model becomes effectively confined to $\mathbb{R}$. 
In this one-dimensional log-gas setting, no such transition occurs (see \cite{DM06}), 
which is consistent with the divergence of the critical value $ \sqrt{2} c_1/c_0^2 \to \infty$ as $c_1 \to \infty.$ 
\end{rem}

\subsection{High dimensional case} \label{S42}

We first determine the equilibrium measure corresponding to $\hat{V}$, assuming that it is fully confined in $\{a\} \times \mathbb{R}^d$. In this setting, the model reduces to a Riesz gas with parameter $s = d-1$ in $\mathbb{R}^d$, subject to the quadratic potential $W(x) = |x|^2$.

For convenience, write  
\begin{equation}
Q(x) =   \frac{ \sqrt{\pi} \, \Gamma(\frac{d+3}{2})   }{ \Gamma(\frac{d}{2}+1)  } \,       |x|^{2}. 
\end{equation}
Then by Theorem \ref{Thm_eq msr for power potential}, the associated equilibrium measure is given by 
\begin{equation}
 d\mu_Q(x)\Big|_{s=d-1} = \frac{ \Gamma( \frac{d+3}{2}) }{ \Gamma(\frac32)\Gamma(\frac{d}{2}+1) } \,  (1-|x|^2)^{ \frac12 }  \cdot  \mathbbm{1}_{ \{ |x| \le 1 \} } \, \frac{ \Gamma(\frac{d}{2}+1) }{\pi^{ d/2 }}\,dx.    
\end{equation}
Note that  
\begin{align*}
|x|^2= R^{-(d-1)} Q(x/R) = R^{-(d+1)}Q(x), 
\end{align*} 
where $R$ is given by \eqref{def of eq msr in a Rd}. 
Then by the change of variables, the equilibrium measure becomes $\mu_W$ in \eqref{def of eq msr in a Rd}. 

We know from (\ref{1.34}) that for $s \in [s-2,d)$, $d \ge 1$ and $|x|<1$,
\begin{equation} \label{3.19}
\int_{ |y|<1 } \frac{ (1-|y|^2)^{\alpha} }{ |x-y|^s }\,dx =  \pi^{d/2}  \frac{ \Gamma(\frac{d-s}{2}) \Gamma(\alpha+1) }{ \Gamma( \frac{d}{2} ) \Gamma(\alpha+1+\frac{d-s}{2}) }   {}_2F_1(\tfrac{s}{2},\tfrac{s-d}{2}-\alpha;\tfrac{d}{2};|x|^2). 
 \end{equation} 
Then as a particular case when $s=d-1$, $\alpha=1/2$, it follows that  
\begin{equation}
\frac{2}{d-1} \int_{ |y|<R } \frac{ d\mu_W(y) }{ |x-y|^{d-1} } = - |x|^2+ \frac{d}{d-1} R^2. 
\end{equation} 
Note here that this identity recovers \eqref{EL condition for semi circle} since  
\begin{align*}
\lim_{d \to 1} \Big( \frac{d}{d-1} R^2- \frac{ 2 }{ d-1 } \Big) = 1+\log 2 .
\end{align*}
 
Recall that  $\hat{x}=(x_0,x_1,\dots,x_d) \in \R^{d+1}$ and $x=(x_1,\dots,x_d) \in \R^d.$
Suppose $x_0=a+t$ with $t\ge0$.
Define 
\begin{align}
\begin{split} \label{def of F(t,x)}
F(t,x)&:= \frac{2}{d-1} \int_{ \R^{d+1} } \frac{ d\mu_V(\hat{y}) }{ |\hat{x}-\hat{y}|^{d-1} } +|\hat{x}|^2   
=\frac{2}{d-1}\int_{ \R^d } \frac{ d\mu_W(y) }{ ( t^2+|x-y|^2 )^{(d-1)/2} }+(a+t)^2+|x|^2. 
\end{split}
\end{align} 
We have shown that  
\begin{equation}
F(0,0)= a^2+  \frac{d}{d-1} R^2, \qquad F(0,x) \begin{cases}
= F(0,0) & \textup{if } |x| \le R,
\smallskip 
\\
\ge F(0,0) & \textup{if } |x| \ge R. 
\end{cases}  
\end{equation}
This is the Euler-Lagrange conditions in the  codimension-one set \( \{a\} \times \mathbb{R}^d \).

We first reduce the $d$-dimensional integral for the potential
to a one-dimensional integral.

\begin{lem} \label{Lem_integral ev for half space constraint}
For any $t \in \R$ and $x \in \R^d$,  
\begin{align}\label{3.24}
\begin{split}
&\quad  \frac{2}{d-1} \int_{ \R^d } \frac{ d\mu_W(y) }{ ( t^2+|x-y|^2 )^{(d-1)/2} }  
\\
&= \frac{4d}{d-1} \frac{R^{d+1}}{\pi}  \int_0^1 \Big( t^2+(|x|+Rr)^2  \Big)^{ -\frac{d-1}{2} } {}_2F_1\Big( \frac{d-1}{2}, \frac{d-1}{2}; d-1;  \frac{ 4 |x| R r }{ t^2+( |x|+R r )^2 }\Big) (1-r^2)^{\frac12} r^{d-1}\,dr .
\end{split}
\end{align} 
In particular, setting $x=0$ in this shows
\begin{equation}\label{3.25}
\frac{2}{d-1}\int_{ \R^d } \frac{d\mu_W(y)}{ (t^2+|y|^2)^{(d-1)/2} }= \frac{2}{d-1}  \frac{1}{t^{d-1} }   {}_2F_1( \tfrac{d-1}{2},\tfrac{d}{2};\tfrac{d+3}{2};-R^2/t^2 ). 
\end{equation}
\end{lem}

\begin{proof}

By applying the Funk-Hecke formula, we have  
\begin{align*}
\begin{split} 
&\quad \int \frac{ d\mu_W(y) }{ ( t^2+|x-y|^2 )^{(d-1)/2} }  = \frac{2R^2}{\pi}   \frac{ \Gamma(\frac{d}{2}+1) }{\pi^{ d/2 }}   \int_{ |y|<1 } \frac{ (1-|y|^2)^{ \frac12 } }{  (t^2/R^2+|x/R-y|^2)^{(d-1)/2}  } \,dy
\\
&= \frac{2R^2d}{\pi}   \frac{ \Gamma(\frac{d}{2}) }{2\pi^{ d/2 }} \int_0^1 \bigg( \int_{ S_1 } \Big( \frac{t^2+|x|^2}{R^2}+r^2-2\frac{|x|}{R}r \frac{x}{|x|} \cdot u \Big)^{ -\frac{d-1}{2} } \,du \bigg) (1-r^2)^{\frac12}r^{d-1}\,dr
\\
&= \frac{2R^2d}{\pi} \frac{ \Gamma( \frac{d}{2} ) }{ \sqrt{\pi}\, \Gamma( \frac{d-1}{2} ) }  \int_0^1 \bigg( \int_{-1}^1 (1-v^2)^{ \frac{d-3}{2} }   \Big( \frac{t^2+|x|^2}{R^2}+r^2-2\frac{|x|}{R}r v \Big)^{ -\frac{d-1}{2} } \,dv \bigg) (1-r^2)^{ \frac12 }r^{d-1}\,dr. 
\end{split}
\end{align*} 
Therefore, we obtain
\begin{align*}
&\quad \int \frac{ d\mu_W(y) }{ ( t^2+|x-y|^2 )^{(d-1)/2} } 
\\
&= \frac{2R^{d+1}d}{\pi} \frac{ \Gamma( \frac{d}{2} ) }{ \sqrt{\pi}\, \Gamma( \frac{d-1}{2} ) }  \int_0^1 \bigg( \int_{-1}^1 (1-v^2)^{ \frac{d-3}{2} }   \Big( (|x|-R r v)^2+ t^2  +R^2r^2(1-v^2)\Big)^{ -\frac{d-1}{2} } \,dv \bigg) (1-r^2)^{ \frac12 }r^{d-1}\,dr. 
\end{align*}
Note that  
\begin{align*}
&\quad \int_{-1}^1   (1-v^2)^{ \frac{d-3}{2} } \Big( \frac{t^2+|x|^2}{R^2}+r^2-2\frac{|x|}{R}r v \Big)^{ -\frac{d-1}{2} }   \,dv  = \int_0^{\pi}  \sin^{ d-2 }(\theta) \Big( \frac{t^2+|x|^2}{R^2}+r^2-2\frac{|x|}{R}r \cos (\theta) \Big)^{ -\frac{d-1}{2} }   \,d\theta 
 \\
 &= 2^{d-2}\int_0^{\pi} \sin^{ d-2 }(\theta/2)\cos^{d-2}(\theta/2)  \Big(  \frac{t^2}{R^2}+\Big(\frac{|x|}{R}+r\Big)^2  -4\frac{|x|}{R}r \cos^2 (\theta/2) \Big)^{ -\frac{d-1}{2} }     \,d\theta  
 \\
 &= \bigg( \frac{t^2}{R^2}+\Big(\frac{|x|}{R}+r\Big)^2  \bigg)^{ -\frac{d-1}{2} } 2^{d-1}  \int_0^{\pi/2}  \sin^{ d-2 }(\theta)\cos^{d-2}(\theta) \Big( 1-X \cos^2(\theta) \Big)^{ -\frac{s}{2} }    \,d\theta  ,
\end{align*}
where 
$$
X=  \frac{ 4 |x| R r }{ t^2+( |x|+R r )^2 }. 
$$
Then we have 
\begin{align*}
&\quad  \int_{-1}^1   (1-v^2)^{ \frac{d-3}{2} } \Big( \frac{t^2+|x|^2}{R^2}+r^2-2\frac{|x|}{R}r v \Big)^{ -\frac{d-1}{2} }   \,dv
\\
&=  \bigg( \frac{t^2}{R^2}+\Big(\frac{|x|}{R}+r\Big)^2  \bigg)^{ -\frac{d-1}{2} } 2^{d-2} \frac{ \Gamma( \frac{d-1}{2} )^2 }{ \Gamma(d-1) } {}_2F_1\Big( \frac{d-1}{2}, \frac{d-1}{2}; d-1;  \frac{ 4 |x| R r }{ t^2+( |x|+R r )^2 }\Big),
\end{align*} 
which completes the proof of (\ref{3.24}). 

Setting $x=0$ in (\ref{3.24}), noting that the ${}_2F_1$ function in the integrand is then unity, and expressing the integral of the remaining terms as a new ${}_2F_1$ gives (\ref{3.25}).
\end{proof}

\begin{lem} \label{Lem_EL for the vertical}
Suppose that $a \ge a_{ \rm cri }.$ Then for any $t \ge 0,$ we have $F(t,0) \ge F(0,0). $  
Furthermore, if $a< a_{ \rm cri }$, this inequality does not hold for $t \in (0,\epsilon)$ for some $\epsilon>0.$
\end{lem}
\begin{proof}
By the previous lemma, the desired inequality is equivalent to show that for any $t \ge 0$,
\begin{equation}
\frac{2}{d-1}  \frac{1}{ t^{d-1} }   {}_2F_1( \tfrac{d-1}{2},\tfrac{d}{2};\tfrac{d+3}{2};-R^2/t^2 ) +2at+t^2  \ge  \frac{d}{d-1} R^2. 
\end{equation}
Note that by \cite[Eq.(15.8.2)]{NIST}, 
\begin{align*}
{}_2F_1( \tfrac{d-1}{2},\tfrac{d}{2};\tfrac{d+3}{2};-R^2/t^2 ) &= \frac{ \sqrt{\pi} \, \Gamma(\frac{d+3}{2})   }{ \Gamma(\frac{d}{2})  } \Big( \frac{t}{R}\Big)^{d-1} {}_2F_1( \tfrac{d-1}{2},-1;\tfrac12; -t^2/R^2 ) 
\\
& \quad -(d^2-1)  \Big( \frac{t}{R}\Big)^{d} {}_2F_1( \tfrac{d}{2},-\tfrac12;\tfrac32; -t^2/R^2 ). 
\end{align*}
Then it follows from \cite[Eq.(15.2.4)]{NIST} that 
\begin{align*}
\frac{2}{d-1}  \frac{1}{ t^{d-1} }  {}_2F_1( \tfrac{d-1}{2},\tfrac{d}{2};\tfrac{d+3}{2};-R^2/t^2 )  
&= d\Big( t^2+\frac{R^2}{d-1} \Big)  -    \frac{  2(d+1)t }{R^d}  {}_2F_1( \tfrac{d}{2},-\tfrac12;\tfrac32; -t^2/R^2 ). 
\end{align*}
Therefore, the desired lemma is equivalent to show that for $a \ge a_{ \rm cri }$ and $t \ge 0$,
\begin{align}
f(t):=(d+1)t+2a  -    \frac{  2(d+1) }{R^d}  {}_2F_1( \tfrac{d}{2},-\tfrac12;\tfrac32; -t^2/R^2 ) \ge 0. 
\end{align}
Note that  
\begin{align}
f(0)= 2(a-a_{ \rm cri }) \ge 0.  
\end{align}
Thus it suffices to show that $f'(t) \ge 0$ for $t \ge 0$.
Since 
\begin{equation}
\frac{d}{dz} {}_2F_1(a,b;c;z)= \frac{ab}{c}\, {}_2F_1( a+1,b+1;c+1;z ), 
\end{equation}
(see \cite[Eq.(15.5.1)]{NIST}), it is equivalent to showing that
\begin{equation}
1-\frac{1}{R^{d+2}}  \frac{2d}{3}\,t \,  {}_2F_1( \tfrac{d}{2}+1, \tfrac12; \tfrac{5}{2}; -t^2/R^2 ) \ge 0,  
\end{equation}
and thus after making use of \eqref{def of eq msr in a Rd} that
\begin{equation}
1-  \frac{ \Gamma(\frac{d}{2}+1)  } { \sqrt{\pi} \, \Gamma(\frac{d+3}{2})   }    \frac{2d}{3}\,s \,  {}_2F_1( \tfrac{d}{2}+1, \tfrac12; \tfrac{5}{2}; -s^2 ) \ge 0.  
\end{equation}
Note that again by \cite[Eq.(15.8.2)]{NIST}, we have
\begin{align*}
\lim_{ s \to +\infty } \Big[  \frac{ \Gamma(\frac{d}{2}+1)  } { \sqrt{\pi} \, \Gamma(\frac{d+3}{2})   }    \frac{2d}{3}\,s \,  {}_2F_1( \tfrac{d}{2}+1, \tfrac12; \tfrac{5}{2}; -s^2 ) \Big] = \frac{d}{d+1}<1. 
\end{align*}
Furthermore, by \cite[Eq.(15.5.3)]{NIST}, we have 
\begin{align*}
\frac{d}{ds}\Big[ s \,  {}_2F_1( \tfrac{d}{2}+1, \tfrac12; \tfrac{5}{2}; -s^2 ) \Big] &=   {}_2F_1( \tfrac{d}{2}+1, \tfrac32; \tfrac{5}{2}; -s^2 ) 
= \frac32 \int_0^1 \frac{ u^{\frac12} }{ (1+s^2u)^{ \frac{d}{2}+1 } }\,du \ge 0. 
\end{align*} 
Therefore the desired inequality follows. 
\end{proof}

\subsection{Large deviation probabilities} \label{S43}

In this subsection, we show the large deviation statement in Theorem~\ref{Thm_Coulomb half space}. 
For the general Riesz gas model, it is well known from \cite[Corollary 1.1]{LS17} that 
\begin{equation}
\log Z_{N,\beta}(V) \sim -\frac{\beta}{2} N^{ 2-\frac{s}{d} } I_V(\mu_V), 
\end{equation}
as $N\to \infty$.  
By definition, the probability $\mathbb{P}[x_0 \ge a]$ can be expressed in terms of the ratio of partition functions
\begin{equation}
  \mathbb{P}[x_0 \ge a] =  \frac{ Z_{N,\beta}(\hat{V}) }{  Z_{N,\beta}(|\hat{x}|^2) }. 
\end{equation}
Therefore, it follows that 
\begin{align}
\log   \mathbb{P}[x_0 \ge a] \sim -\frac{\beta}{2} N^{ \frac{d+3}{d+1} }  \Big( I_V(\mu_V) - I_{ |x|^2 }(\mu_{|x|^2}) \Big),
\end{align}
as $N \to \infty.$

Note that by \eqref{energy for power type potential} with $p=1$, $d \mapsto d+1$ and $s=d-1$, we have 
\begin{equation}
I_{ |x|^2 }(\mu_{|x|^2}) = \frac{(d+1)^2}{(d-1)(d+3)}. 
\end{equation} 
On the other hand, we have 
\begin{align*}
I_V(\mu_V)  = a^2+I_W(\mu_W). 
\end{align*}
Note that by \eqref{3.19}, the Robin's constant associated with $W$ is given by $\frac{d}{d-1}R^2$. Therefore by \eqref{energy in terms of Robin}, after straightforward computations, we obtain
\begin{equation}
I_V(\mu_V)= a^2+ \Big( \frac{ \sqrt{\pi} \, \Gamma(\frac{d+3}{2})   }{ \Gamma(\frac{d}{2}+1)  } \Big)^{ \frac{2}{d+1} }\frac{d(d+1)}{(d-1)(d+3)}. 
\end{equation}
Thus it follows that 
\begin{align}
\begin{split}
I_V(\mu_V) - I_{ |x|^2 }(\mu_{|x|^2}) &= a^2+ \Big( \frac{ \sqrt{\pi} \, \Gamma(\frac{d+3}{2})   }{ \Gamma(\frac{d}{2}+1)  } \Big)^{ \frac{2}{d+1} }\frac{d(d+1)}{(d-1)(d+3)}-  \frac{(d+1)^2}{(d-1)(d+3)} . 
\end{split}
\end{align}
Combining the above, we obtain \eqref{LD Prob in main}.

\appendix 

\section{Conjecture~\ref{Cj1} for $d=3$} \label{Appendix_d=3}
In this appendix, we verify Conjecture~\ref{Cj1} in the case \( d=3 \).
As explained in Remark~\ref{Rem_conjecture inequality}, it is sufficient to
establish that \( G(x) \ge G(0) \), where \( G(x) \) is defined in
\eqref{def of G(x)}.

By \cite[Eq.(15.8.10)]{NIST}, we have 
\begin{align} \label{hypergeometric in G(x)}
{}_2F_1(a,a;2a;z)= -\frac{ \Gamma(2a) }{\Gamma(a)^4} \sum_{k=0}^\infty \frac{ \Gamma(a+k)^2 }{ (k!)^2 } (1-z)^k \Big( \log(1-z) +2 \psi(k+a)- 2 \psi(k+1) \Big),  
\end{align}
where $\psi(z)=\Gamma'(z)/\Gamma(z)$ is the digamma function. In particular for $a=1$, we have 
\begin{align}
{}_2F_1(1,1;2;z) = - \frac{\log (1-z)}{z}. 
\end{align}
By using this, the function $G(x)$ in \eqref{def of G(x)} for $d=3$ is simplified as 
\begin{align}
\begin{split}
G(x) &  = -\frac{3}{2 \pi x}\int_0^1  r \sqrt{1-r^2} \log \Big(\frac{(r-x)^2+t^2}{(r+x)^2+t^2}\Big) \, dr + x^2
\\
&= -\frac{3}{ \pi x} \re \int_0^1  r \sqrt{1-r^2} \log \Big(\frac{r-x-it}{r+x-it}\Big)\, dr + x^2
\\
&= -\frac{3}{ \pi x} \re \int_{-1}^1  r \sqrt{1-r^2} \log (r-x-it)\, dr + x^2.
\end{split}
\end{align} 
Note that by using the Stieltjes transform \eqref{3.16+} of the semi-circle law, we have 
\begin{align*}
(xG(x))'&= 3x^2 + \frac{3}{\pi}\re \int_{-1}^1 \frac{r}{r-x-it} \sqrt{1-r^2} \,dr  = 3x^2 + \frac{3}{\pi}\re \int_{-1}^1 \Big(1+ \frac{x+it}{r-x-it} \Big) \sqrt{1-r^2} \,dr 
\\
&= 3x^2+ \frac{3}{2} + 3 \,\re \bigg[ (x+it) \int_{-1}^1 \frac{1}{r-x-it} \frac{\sqrt{1-r^2}}{\pi}\,dr \bigg]  
\\
&= 3x^2+ \frac{3}{2} + 3 \,\re \Big[ -z^2+z\sqrt{ z^2-1 }  \Big]  = 3t^2+\frac32 +3 \re \Big[ z \sqrt{z^2-1} \Big],
\end{align*}
where $z=x+it$. 
Therefore it follows that 
\begin{align} \label{G(x) for d=3}
G(x) = 3t^2+\frac32 + \frac{1}{x}  \re \Big[(z^2-1)^{\frac32}\Big] .  
\end{align}
Let $w=\sqrt{z^2-1}=u+iv$. Then we have $u^2-v^2=x^2-t^2-1$ and $uv=tx$. 
Thus we have
\begin{align*}
x^2G'(x) =  3x \re (z w) -\re w^3 =3x^2u-3xtv-u^3+3uv^2 = u(3x^2-u^2). 
\end{align*}
Since $u>0$, it suffices to show that $3x^2-u^2>0$. 
By using $(u^2+v^2)^2= (x^2-t^2-1)^2+ 4 t^2x^2 $, we have 
$$
u^2= \frac{ \sqrt{ (x^2-t^2-1)^2+ 4 t^2x^2} +x^2-t^2-1 }{2}  . 
$$
Then the desired inequality $3x^2-u^2>0$ is equivalent to 
$$
5x^2+t^2+1 > \sqrt{ (x^2-t^2-1)^2+ 4 t^2x^2} . 
$$
This inequality immediately follows from 
$$
(5x^2+t^2+1 )^2-\Big(  (x^2-t^2-1)^2+ 4 t^2x^2 \Big)= 4x^2(3+2t^2+6x^2)>0. 
$$
Therefore, we obtain that $G'(x) \ge 0$. 

\begin{rem} 
The function $G(x)$ in \eqref{G(x) for d=3} admits an explicit representation in terms of $x$ and $t$: 
\begin{equation}
G(x) =   3t^2 + \frac{3}{2}-\frac{1}{x} \Big( (t^2-x^2+1)^2+ 4 t^2 x^2\Big)^{ \frac{3}{4} } \sin \bigg(\frac{\pi}{4}+ \frac{3}{2} \arctan\Big(\frac{t^2-x^2+1}{2
   t x}\Big)\bigg). 
\end{equation}
From this expression, one can also directly verify that $G'(x) \ge 0.$
\end{rem}

\begin{rem}
The approach used in this appendix for \( d=3 \) may potentially be extended
to odd dimensions \( d \), since, by \eqref{hypergeometric in G(x)}, the
hypergeometric function \( {}_2F_1 \) appearing in \eqref{def of G(x)} again
reduces to a combination of logarithm and rational functions. However,
carrying out such computations in a systematic manner remains challenging,
and we leave this for future work. 
\end{rem}


\bibliographystyle{abbrv}

\end{document}